\documentclass[runningheads]{llncs}
\usepackage{graphicx}
\def\BibTeX{{\rm B\kern-.05em{\sc i\kern-.025em b}\kern-.08em
		T\kern-.1667em\lower.7ex\hbox{E}\kern-.125emX}}

\usepackage{booktabs} % For formal tables
\usepackage{times,amsmath,epsfig}
\usepackage{url}
\usepackage{multirow}
\usepackage{subfig}
\usepackage{graphicx}
\usepackage{xspace}
\usepackage[T1]{fontenc}
\usepackage{amssymb} % assumes amsmath package installed
\usepackage{algorithm}
\usepackage{algorithmic}

\newcommand{\btitle}[1]{\noindent\underline{\textbf{#1}}}

\newtheorem{observation}{Observation}

\newcommand{\ie}{\emph{i.e.,}\xspace}
\newcommand{\eg}{\emph{e.g.,}\xspace}
\newcommand{\resp}{\emph{resp.,}\xspace}
\newcommand{\etal}{\emph{et al.}\xspace}

\newcommand{\eat}[1]{}

\def\EndOfProof{\nolinebreak\ \hfill\rule{2.0mm}{2.0mm}}
\begin{document}
\title{An Efficient Secure Dynamic Skyline Query Model\thanks{corresponding author: Hui Li, hli@xidian.edu.cn}}
\author{Weiguo Wang\inst{1}\and
	Hui Li\inst{1}\orcidID{0000-0003-2382-6289} \and
	Yanguo Peng\inst{2}\and
	Sourav S Bhowmick\inst{3}\and
	Peng Chen\inst{1}\and
	Xiaofeng Chen\inst{1}\and
	Jiangtao Cui\inst{2}}
%%\orcidID{2222--3333-4444-5555}
%\authorrunning{W. Wang, H. Li et al.}
%\titlerunning{Secure Dynamic Skyline Query}
%% First names are abbreviated in the running head.
%% If there are more than two authors, 'et al.' is used.
%%
\institute{School of Cyber Engineering, Xidian University, China 
	\email{\{wgwang,pchen97\}@stu.xidian.edu.cn, \{hli,xfchen\}@xidian.edu.cn}\and 
	School of Computer Science and Technology, Xidian University, China
	\email{\{ygpeng,cuijt\}@xidian.edu.cn}\and
	School of Computer Science and Engineering, Nanyang Technological University, Singapore
	\email{assourav@ntu.edu.sg}}
%\titlerunning{Abbreviated paper title}
% If the paper title is too long for the running head, you can set
% an abbreviated paper title here
%
%\author{Weiguo Wang\inst{1}\orcidID{0000-1111-2222-3333} \and
%Hui Li\inst{2}\orcidID{1111-2222-3333-4444} \and
%Yanguo Peng\inst{3}\orcidID{2222--3333-4444-5555}\and
%Sourav S Bhowmick\inst{4}\orcidID{0001-1111-2222-3333}}
%
%\authorrunning{F. Author et al.}
% First names are abbreviated in the running head.
% If there are more than two authors, 'et al.' is used.
%
%\institute{Princeton University, Princeton NJ 08544, USA \and
%Springer Heidelberg, Tiergartenstr. 17, 69121 Heidelberg, Germany
%\email{lncs@springer.com}\\
%\url{http://www.springer.com/gp/computer-science/lncs} \and
%ABC Institute, Rupert-Karls-University Heidelberg, Heidelberg, Germany\\
%\email{\{abc,lncs\}@uni-heidelberg.de}}
%
\maketitle              % typeset the header of the contribution
\begin{abstract}
It is now cost-effective to outsource large dataset and perform query over the cloud. However, in this scenario, there exist serious security and privacy issues that sensitive information contained in the dataset can be leaked. The most effective way to address that is to encrypt the data before outsourcing. Nevertheless, it remains a grand challenge to process queries in ciphertext efficiently. In this work, we shall focus on solving one representative query task, namely \textit{dynamic skyline query}, in a secure manner over the cloud. However, it is difficult to be performed on encrypted data as its dynamic domination criteria require both subtraction and comparison, which cannot be directly supported by a single encryption scheme efficiently.
To this end, we present a novel framework called \textsc{scale}. It works by transforming traditional dynamic skyline domination into pure comparisons. The whole process can be completed in single-round interaction between user and the cloud. We theoretically prove that the outsourced database, query requests, and returned results are all kept secret under our model. Moreover, we also present an efficient strategy for dynamic insertion and deletion of stored records. Empirical study over a series of datasets demonstrates that our framework improves the efficiency of query processing by nearly\textbf{ three orders of magnitude} compared to the state-of-the-art.

\keywords{skyline, secure, cloud, query}
\end{abstract}
\section{Introduction}\label{sec:intro}
With the rapid expansion in data volumes, many individuals and organizations are increasingly inclined to outsource their data to public cloud services since they provide a cost-effective way to support large-scale data storage and query processing.
As a major type of query and fundamental building block for various applications, \emph{skyline query}~\cite{borzsony2001skyline} has become an important issue in database research for extracting interesting objects from multi-dimensional datasets. The skyline query processing is
widely adopted in many applications that require multi-criteria decision making such as market research~\cite{DBLP:conf/vldb/DellisS07}, location based systems~\cite{DBLP:conf/icde/KriegelRS10}, web services study~\cite{DBLP:conf/www/AlrifaiSR10}, etc. The \textit{skyline operator} filters out a set of interesting points based on a group of evaluation criteria from a large set of points. A point is considered as interesting, if there does not exist a point that is at least as good in all criteria and better in at least one criteria.

However, similar to other types of query, outsourcing skyline query workload to a public cloud will inevitably raise privacy issues. Since a real-world database may often contain sensitive information such as personal electronic mails, health records, financial transactions, etc., a cloud service provider may illegally spy on the data and invade the privacy of the data owner and users.

In this paper, we focus on the problem of \emph{secure} skyline querying on the cloud platform aiming to protect the security of outsourced data, query request and results. Secure query processing on encrypted data is an important research field in outsourcing computation and has been extensively studied during recent years~\cite{DBLP:conf/infocom/SunZLH17,Liu2017Secure}. For instance, fully homomorphic encryption schemes \cite{gentry2009fully} ensure strong security while enabling arbitrary computations on the encrypted data. Modular Order-preserving encryption \cite{boldyreva2011order,DBLP:conf/asiacrypt/ChatterjeeD15} provides an intuitive security model which supports comparison over the ciphertext without decryption. Despite the promising achievements in the area of secure query processing, it remains a grand challenge for processing \emph{dynamic} skyline queries over ciphertext, where the \textit{skyline operator} is executed with respect to some query point~\cite{DBLP:conf/sigmod/PapadiasTFS03}. The main reason for the problem is as follows. Given a query request, a dynamic skyline query requires performing \emph{both} comparison and distance evaluation online simultaneously. Unfortunately, accomplishing this task over ciphertext cannot be realized efficiently via existing encryption schemes.

For instance, suppose that a medical institution wishes to outsource its electronic diabetes records to some public cloud service, resisting to leak the content of the records to the cloud server. An electronic diabetes record consists of a series of attributes, including \textit{ID}, \textit{age}, \textit{FBGL (fasting blood glucose level)}, etc. Let $P=\{p_1,\ldots,p_n\}$ denote a set of electronic diabetes records. When the medical institution receives a new record $q$, it expects the cloud server to retrieve a similar record to enhance and personalize the treatment for the new patient $q$. However, it is usually difficult or even impossible to uniformly assign weights for all the attributes to return the nearest neighbor (\eg $p_i$ is the nearest if only \textit{age} is involved while $p_j$ is the nearest if only \textit{FBGL} is taken into account). In light of that, dynamic skyline query provides all possible Pareto records that are not dominated by any other ones. Given a query $q$, we can compute the difference between each attribute for $p_i$ and $q$. Let $t_i$ be the difference tuple between $p_i$ and $q$, and $t_i[j] = |p_i[j] - q[j]|$ for each dimension $j$. An object $t_i$ dominates $t_j$ if it is better than $t_j$ in at least one dimension and not worse in every other dimensions. If an object cannot be dominated by any other object, this object is one of the skyline points that needs to be returned. As shown in Fig.~\ref{skylineexample}, there are five patient records $p_1,\ldots,p_5$. Given a query record $q$, we calculate $t_1,\ldots,t_5$ and can easily identify the skyline points as $t_5$ and $t_2$. Therefore, $p_5$ and $p_2$ are the results for the dynamic skyline query w.r.t $q$.

\begin{figure}[t]
	\center
	\includegraphics[width=0.6\columnwidth]{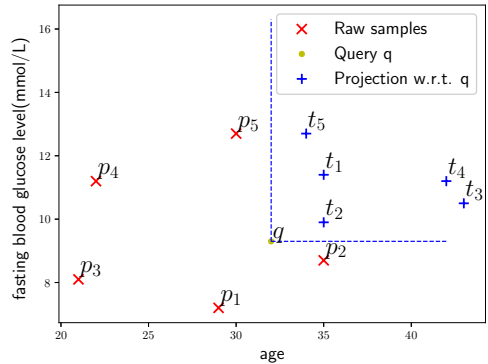}
	\vspace{-2ex}\caption{Dynamic Skyline Query Example.}
	\label{skylineexample}\vspace{-3ex}
\end{figure}

Notably, in the above example, a dynamic skyline query requires performing both subtraction and comparison online. As there is no practical encryption scheme supporting both operators over ciphertext, existing model employs secure multiparty computation over at least to third-party non-collusion clouds and processes the query with multiple rounds of interactions. In this work, we present a novel framework called \textsc{scale} (\textbf{S}e\textbf{C}ure dyn\textbf{A}mic sky\textbf{L}ine qu\textbf{E}rying) by transforming traditional skyline domination criteria, which requires both subtraction and comparison, into comparison only. In this way, we are able to present a new scheme that can support dynamic skyline query over ciphertext without any help from a second cloud and can be completed in a single-round interaction between user and the cloud. We theoretically prove that the outsourced database, query requests, and returned results are all kept secret under our model. Empirical study over four datasets including both synthetic and real-world ones demonstrate that our framework outperforms a state-of-the-art method by nearly three orders of magnitude. Notably, as a special case of dynamic skyline query, skyline computation can also be processed securely and efficiently under our model (with trivial modification).

In summary, this work makes the following contributions.
\begin{itemize}
	\item We propose a new scheme to encrypt the outsourced database and query request. Based on the scheme, dynamic skyline query can be answered without decrypting the database or the query. Within the scheme, the cloud server and data user need only one interaction during the query.
	\item We theoretically prove that our model is secure if the cloud is curious-but-honest.
	\item In addition to the secure query scheme, we also present an efficient mechanism for modifications over existing database records, including insertion, deletion and updating.
	\item We also theoretically show that the skyline points can be computed efficiently and correctly.
	\item Empirical study over both synthetic and real-world datasets justify that our model is superior to the state-of-the-art w.r.t the query response time, by more than \textbf{\textit{3 orders of magnitude faster}}.
\end{itemize}

The rest of this paper is organized as follows. In Section \ref{sec:RelatedWork} we conduct a literature review for the related work.
In Section \ref{sec:PROBLEMDEFINITION} we formally present the problem definition and system model for this work. The detailed designs of encryption scheme and query framework are discussed in Section \ref{sec:SecureSkyline}. Empirical study and corresponding results are shown in Section \ref{sec:Experiments}. In Section
\ref{sec:Conclusion}, we conclude this work.

\vspace{0ex}\section{Related Work}
\label{sec:RelatedWork}
\vspace{0ex}\subsection{Skyline Query in Plaintext}

The skyline query is particularly important for several applications
involving multi-criteria decision making. The computation of the skyline is equivalent to determining the maximal vector
problem in computational geometry~\cite{DBLP:journals/jacm/KungLP75,compGeo1985}, or equivalently the Pareto optimal set~\cite{compGeo1985} problem in
operations research.
Since \cite{DBLP:journals/jacm/KungLP75} earliest studied the method and complexity of skyline computation (\ie \textit{static} skyline query)\footnote{In the following of this work, we shall refer to static skyline query as skyline computation.},
it has been extensively studied in the database field. Block-Nested-Loop~\cite{borzsony2001skyline}, Divide-and-Conquer~\cite{Chomicki2003Skyline},
Nearest Neighbor~\cite{Kossmann2002Shooting}, Branch-and-Bound Skyline~\cite{DBLP:conf/sigmod/PapadiasTFS03} and series of works afterwards have progressively improved the efficiency on general version of skyline computation~\cite{mullesgaard2014efficient,park2017efficient,zhou2016adaptive}.\eat{A static version of skyline~\cite{}, given a set of multi-dimensional data points, skyline computation returns the set of skyline points in the dataset that are not dominated by any others.}

A \textit{dynamic} skyline query is a variation of skyline computation that was first introduced in~\cite{DBLP:conf/sigmod/PapadiasTFS03,DBLP:journals/tods/PapadiasTFS05}. Instead of computing the skyline points purely from the given dataset, dynamic skyline query returns series of points that are not dominated by any others with respect to $q$. In another word, skyline computation can be viewed as a special case of dynamic skyline query where $q$ is fixed as origin point and only the comparison (without distance evaluation) is required.

\vspace{0ex}\subsection{Secure Skyline Query over Encrypted Data}

With the development of cryptography, Encryption technology is gradually applied in the database field.

Bothe \etal\cite{Bothe2014Skyline} presented an approach for skyline computation over Encrypted Data. It provided efficiency analysis and empirical study for computing skyline points and decrypting the results. However, it failed to provide any formal security guarantee. Moreover, as discussed above, skyline computation only requires comparison without distance computation; processing it in encrypted domain can be easily performed through Order Preserving Encryption~\cite{agrawal2004order} or Order Revealing Encryption~\cite{boldyreva2011order,DBLP:conf/asiacrypt/ChatterjeeD15}. Another work~\cite{Chen2016Secure}
proposed three novel schemes that enable efficient verification of skyline query results returned by an unauthentic cloud server. Their work focuses on the verification but not privacy issues, and does not work on ciphertext. It is orthogonal to the scope of this paper.

Liu \etal\cite{Liu2017Secure} proposed the first semantically secure protocol for dynamic skyline query over the cloud platform. The scheme adopts both Paillier cryptosystem \cite{paillier1999public} and Secure Multi-party Computation (SMP) as building blocks. Although it is proved to be semantically secure, the protocol suffers from huge computation cost and strict system model. In fact, as a query framework, the response time is the most important issue for the success of the application, but the performance of \cite{Liu2017Secure} is far from satisfactory in this aspect.

\vspace{-1ex}\subsection{Order-Revealing Encryption}

Order-Preserving Encryption (OPE) scheme~\cite{agrawal2004order}, whose ciphertext preserve the original ordering of the plaintexts, has been extensively applied in range query over encrypted databases.

The ideal security goal for an order-preserving scheme, IND-OCPA~\cite{boldyreva2009order}, is to reveal no additional information about the plaintext values besides their order. Boldyreva \etal\cite{boldyreva2009order} were the first to provide a rigorous solution to the problem. They settled on a weaker security guarantee, which was shown to leak at least half of the plaintext bits ~\cite{boldyreva2011order,DBLP:conf/asiacrypt/ChatterjeeD15}. Popa \etal\cite{popa2013ideal} presented the ideal-secure order-preserving encoding scheme. \cite{kerschbaum2015frequency} showed how to achieve the even stronger notion of frequency-hiding OPE. However, these ideal-secure OPE schemes require rounds of interactions between client and server.

To improve the security of OPE, Boneh \etal\cite{Dan2015Semantically} presented Order-Revealing Encryption schemes (ORE), another method for circumventing the lower bound deduced by Boldyreva \etal \cite{boldyreva2009order}. In an ORE scheme, the numerical order of two ciphertexts does not necessarily reflect that of the original messages as OPE does. Instead, the order of the original messages can only be decided by a carefully designed function over the corresponding ciphertexts. Pandey and Ruselakis \cite{Pandey2012Property} previously considered this type of relaxation in the context of \textit{property-preserving encryption}. In a property-preserving encryption scheme, there is a publicly computable function that can be evaluated on ciphertexts to determine the value of some property on the underlying plaintexts. OPE can thus be viewed as a property-preserving encryption scheme where the computable function is the comparison operation. Pandey and Rouselakis introduced and explored several indistinguishable-based notions of security for property-preserving encryption. However, they did not construct an order-revealing encryption scheme. Chenette \etal\cite{Chenette2016Practical} built efficiently implementable order-revealing encryption based on pseudorandom functions. Lewi \etal\cite{Lewi2016Order} improved the above scheme. The ORE scheme in \cite{Lewi2016Order} is adopted for this work, and it will be discussed further in Section \ref{sec:SecureSkyline}.

\vspace{0ex}\section{Problem Definition}
\label{sec:PROBLEMDEFINITION}
In this section, we shall introduce a group of key concepts for skyline query, and finally describe the system and security models in this paper. For ease of discussion, the key notations used throughout this paper are summarized in Table \ref{tb:notation}.

\begin{table}[t]
	\caption{Summary of notations}\vspace{0ex}
	\begin{center}
		\begin{tabular}{l|p{5.9cm}}
			\toprule
			\textbf{Notation} & \textbf{Definition} \\
			\midrule
			$n$	& The number of tuples in the database \\
			$d$ & The dimension of database \\
			$q$ & The query tuple \\
			$\mathsf{Enc}(q)$ & Ciphertext of the query tuple \\
			$\mathsf{Enc}(2q)$ & Ciphertext of the doubled query tuple \\
			$P=\{p_1,\ldots,p_n\}$ & A database with $n$ tuples \\
			$E(P)$ & Ciphertexts of tuples for $P$ \\
			$E(\Phi)$ & Ciphertexts of the pairwise sums for tuples in $P$ \\
			$p_i[j]$ & The $j-th$ attribute of $p_i$ \\
			$keys[\cdot]$ & The set of private keys \\
			\bottomrule
		\end{tabular}\label{tb:notation}
	\end{center}\vspace{0ex}
\end{table}

\vspace{0ex}\subsection{Skyline Query Definition}
In this part, we shall introduce a series of key concepts for skyline problem that is important for the following discussion.
\begin{definition}
	[Domination] Given two
	points $p_\alpha$, $p_\beta$ in $d$-dimensional space, we say $p_\alpha$
	dominates $p_\beta$ (denoted by $p_\alpha\prec p_\beta$), if $\forall i \in \{1,\ldots,d\}$, $p_\alpha[i] \leq p_\beta[i]$,
	and $\exists i \in \{1,\ldots,d\}$, $p_\alpha[i] < p_\beta[i]$.
\end{definition}\label{def:d}
\begin{definition}
	[Skyline Computation] Given a dataset $P=\{p_1,\ldots,p_n\}$ in $d$- dimensional space,
	skyline computation returns the points set $S\subseteq P$, such that $\forall p\in S$, $\nexists p'\in P$ such that $p'\prec p$ (\ie $\forall p \in S, p' \in P$, $p'$ cannot dominate $p$).
\end{definition}\label{def:sc}
\begin{definition}
	[Dynamic Domination] Given two
	points $p_\alpha$, $p_\beta$ and a query point $q$ in $d$-dimensional space, we say $p_\alpha$
	dynamically dominates $p_\beta$ with respect to $q$ (denoted by $p_\alpha\prec_q p_\beta$), if $\forall i \in \{1,\ldots,d\}$, $|p_\alpha[i]-q[i]| \leq |p_\beta[i]-q[i]|$,
	and $\exists i \in \{1,\ldots,d\}$, $|p_\alpha[i]-q[i]| < |p_\beta[i]-q[i]|$.
\end{definition}\label{def:dd}
\begin{definition}
	[Dynamic Skyline Query] Given a dataset $P=\{p_1,\ldots,p_n\}$ and a query $q$ in $d$-dimensional space, dynamic skyline query returns the set $S\subseteq P$, such that $\forall p\in S$, $\nexists p'\in P$ such that $p'\prec_q p$ (\ie $\forall p \in S, p' \in P$, $p'$ cannot dynamically dominate $p$ with respect to $q$).
\end{definition}\label{def:ddsq}

A common algorithm (\ie BNL~\cite{borzsony2001skyline}) for dynamic skyline query is shown in Algorithm \ref{skylinequery}. It first calculates the differences (\ie $t_i$) between each tuple (\ie $p_i$) and the query request (\ie $q$) in every dimension (Lines 1-3). When a tuple $p_i$ is read from $P$, it is added to $S$ if $S$ is empty (Lines 5-6). Otherwise, we shall compare $p_i$'s corresponding difference tuple with respect to $q$, namely $t_i$, with that of each tuple in $S$. In case $t_i\prec t_j$, where $p_j\in S$, we shall delete $p_j$ from $S$. If there is no $p_j\in S$ such that $t_j\prec t_i$, we shall add $p_i$ to $S$ (Lines 10-11, 16-18). The algorithm repeats this process for the remaining tuples in $P$, and finally returns $S$ (Line 21).

\begin{algorithm}[t]
	\caption{Basic Skyline Query Algorithm}
	\label{skylinequery}
	\begin{algorithmic}[1]
		\REQUIRE The dataset $P$ and a query tuple $q$
		\ENSURE The result set of skyline points $S$
		\FOR{$i$ in $1,\ldots,n$ and $j$ in $1,\ldots,d$}
		\STATE let $t_i[j] = |p_i[j] - q[j]|$
		\ENDFOR
		\FOR{$i$ in $1,\ldots,n$}
		\IF{$S$ is empty}
		\STATE add $p_i$ to $S$
		\ELSE
		\STATE $flag \leftarrow True$
		\FOR{each $p_j\in S$}
		\IF{$t_j\prec t_i$}
		\STATE $flag \leftarrow False$
		\ELSIF{$t_i\prec t_j$}
		\STATE delete $p_j$ from $S$
		\ENDIF
		\ENDFOR
		\IF{$flag == True$}
		\STATE add $p_i$ to $S$
		\ENDIF
		\ENDIF
		\ENDFOR
		\RETURN $S$
	\end{algorithmic}
\end{algorithm}
We shall use this as the basis for our secure skyline model. Notably, this is not the most efficient algorithm for plaintext skyline query. We select this method as our building block for the following reasons. Firstly, the state-of-the-art solution for secure dynamic skyline is~\cite{Liu2017Secure}, it adopts BNL~\cite{borzsony2001skyline} as the basic building block. In line with them and to make a fair comparison, our solution is constructed according to the same query framework. Secondly, BNL is a common and popular iterative algorithm for answering dynamic skyline query in plaintext. Thirdly, as discussed in Section~\ref{sec:intro}, the key challenge in secure dynamic skyline query lies in the solution for performing both subtraction and comparison over ciphertext. A secure model building on any other (plaintext) dynamic skyline query algorithm inevitably has to address that. In other words, although our solution in this work adopts Algorithm \ref{skylinequery} as the foundation, it can be easily adapted to other (plaintext) dynamic skyline query algorithms.

%First, according to the discussion in~\cite{Liu2017Secure}, compared to recent efforts in divide-and-conquer-based algorithms, this iterative approach is easy to be securely implemented and is less possible to suffer from the ``curse-of-dimensionality''. Second, as discussed in Section~\ref{sec:intro}, the key challenge in secure dynamic skyline query lies in the solution for performing both subtraction and comparison over ciphertext. A secure model building on any other (plaintext) dynamic skyline query algorithm inevitably has to address that. In other words, although our solution in this work adopts Algorithm \ref{skylinequery} as the foundation, it can be easily adapted to other (plaintext) dynamic skyline query algorithms. Third, BNL~\cite{borzsony2001skyline} is a common and popular iterative algorithm for answering dynamic skyline query in plaintext. Last but not the least, to make a fair comparison, in line with~\cite{Liu2017Secure}, which is the only and latest solution for the same task, we adopt BNL~\cite{borzsony2001skyline} as the basic query algorithm in this work.

\vspace{0ex}\subsection{System Model and Design Goals}\label{ssec:secmodel}
Our system model involves three types of participants: a data owner, a cloud server and a group of query users. The cloud server is assumed to have large storage and computation ability, and it provides outsourcing storage and computation services. As Fig.~\ref{model} shows, the data owner employs the cloud service and stores his private database in the cloud server. To preserve data privacy, the data owner will encrypt his dataset, and only outsource the encrypted dataset to the cloud. Every query user may submit a query point (\ie $q$) toward the system. The query request may be locally encrypted before sending to the cloud server. Then, the cloud server will perform dynamic skyline query over encrypted database and query request without decryption. Afterwards, it returns the encrypted results to the user. Finally, the user decrypts these results using their own private keys.

\begin{figure}[t]
	\center
	\includegraphics[width=0.75\columnwidth]{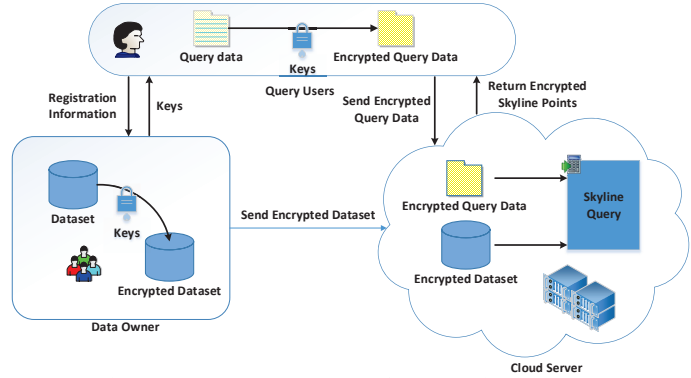}
	\vspace{-3ex}\caption{The system model of secure skyline query}
	\label{model}
	\vspace{0ex}\end{figure}

\vspace{0ex}\subsubsection{Security model}
We parameterize the security model by a collection of leakage functions
\[\mathcal{L} = (\mathcal{L}_{Encrypt}, \mathcal{L}_{Query}, \mathcal{L}_{Insert}, \mathcal{L}_{Delete}).\]
The functions describe what information the protocol leaks to the adversary. The definition ensures that the scheme does not reveal any information beyond what can be inferred from the leakage functions.

We define two games ${\rm Game}_{\mathcal{R},\mathcal{A}}$ and ${\rm Game}_{\mathcal{S},\mathcal{A}}$ as follows. The adversary repeatedly encrypts data and queries skyline points, and receives the transcripts generated from $Encrypt()$ and $Query()$ algorithms in the real game ${\rm Game}_{\mathcal{R},\mathcal{A}}$ or receives the transcripts generated by the simulator $\mathcal{S}(\mathcal{L}_{Encrypt})$ and $\mathcal{S}(\mathcal{L}_{Query})$ in the ideal game ${\rm Game}_{\mathcal{S},\mathcal{A}}$. Eventually, $\mathcal{A}$ outputs a bit 0 (${\rm Game}_{\mathcal{R},\mathcal{A}}$) or 1 (${\rm Game}_{\mathcal{S},\mathcal{A}}$).
\begin{definition}[Adaptively secure]
	A scheme is $\mathcal{L}$-adaptively-secure if for all probabilistic polynomial-time algorithm $\mathcal{A}$, there exists an efficient simulator $\mathcal{S}$ such that the following equation holds:
	\[\big|Pr[{\rm Game}_{\mathcal{R},\mathcal{A}}(\lambda)=1] - Pr[{\rm Game}_{\mathcal{S},\mathcal{A}}(\lambda)=1]\big| \leq negl(\lambda).\]
	\label{df:AdaptivelySecure}
\end{definition}\vspace{0ex}
\vspace{0ex}\subsubsection{Design goals}
Our design goals contain both efficiency and privacy, including database privacy, query privacy, and result privacy. The details are as follows.
\begin{itemize} 
	\item Data owners need to encrypt the database before it is sent to the cloud server. Meanwhile, the content in the database is not leaked to the cloud server.
	\item Query request, as well as the results, should not be revealed to the cloud server throughout query processing.
	\item As a query processing framework, efficiency should be considered as one of the most important issue for measuring its success. Although the entire query processing is performed in ciphertext here, it should minimize the additional cost associated with it.
\end{itemize}

\vspace{0ex}\section{The SCALE Framework}
\label{sec:SecureSkyline}
In this section, we shall introduce the \textsc{scale} framework for secure dynamic skyline querying under the proposed system model. As discussed above, processing dynamic skyline query given a query point $q$ requires performing both subtraction and comparison. Addressing both tasks in ciphertext form is challenging as there is no practical encryption scheme that supports both operations simultaneously.

To address this challenge, we reinvestigate the entire dynamic skyline query workflow described in Definition~\ref{def:ddsq} and Algorithm~\ref{skylinequery}. Our investigation revealed an important fact that may lead to an effective solution. Notably, to answer a dynamic skyline query given a request $q$, quantifying the differences between each point $p_i$ and $q$ through all dimensions is not mandatory. Instead, what we need is the relative order of such differences for a group of different $p_i$. 

\begin{observation}\label{obs1}
	In order to evaluate whether $p_\alpha$ dynamically dominates $p_\beta$ with respect to $q$, we do not need to know the exact values for the difference vectors $T_\alpha$ and $T_\beta$, where $T_i[j]=|p_i[j]-q[j]|$ for $j\in[1,\ldots,d]$. In fact, what we really need to know is whether $T_\alpha[j] \le T_\beta[j]$ or $T_\alpha[j] < T_\beta[j]$ for $j\in[1,\ldots,d]$. For simplicity, for an arbitrary dimension $j$, we need to know whether $p_\alpha[j]$ or $p_\beta[j]$ is close to $q[j]$. To answer that, we have to consider two possible cases depending on whether $q[j]$ falls in the interval between $p_\alpha[j]$ and $p_\beta[j]$. Fig.~\ref{compare1} and Fig.~\ref{compare2} depict the cases. In the case of Fig.~\ref{compare1}, the order between $T_\alpha[j]$ and $T_\beta[j]$ can be interpreted as the relationship between $p_\alpha[j]$ and $p_\beta[j]$. In the case of Fig.~\ref{compare2}, the order between $T_\alpha[j]$ and $T_\beta[j]$ can be interpreted as the relationship between $p_\alpha[j]+p_\beta[j]$ and $q[j]+q[j]$.
\end{observation} 

In the aforementioned study, we notice that the multi-type-operation requirement (\ie with both subtraction and comparison) in dynamic skyline query can be transformed to \textbf{uni-type-operation involving only comparison}. Inspired by this critical point, current encryption schemes that support comparison over ciphertext can be adopted in our framework to realize our design goals.

\vspace{0ex}\subsection{Database Encryption}\label{ssec:41}
In our scheme, we adopt a state-of-the-art encryption scheme that supports comparison, namely \textit{order-revealing encryption} \cite{Lewi2016Order}. We first present the formal definition of order-revealing encryption.
\begin{definition}
	[Order-Revealing Encryption] An order-revealing encryption (ORE) scheme \cite{Lewi2016Order}
	is a tuple of three algorithms including $\mathsf{Setup}$, $\mathsf{Encrypt}$ and $\mathsf{Compare}$ defined over a
	well-ordered domain $D$ with the following properties:
	\begin{itemize} 
		\item $\mathsf{Setup}(1^\lambda)\rightarrow sk$: On input a security parameter $\lambda$, the setup
		algorithm outputs a secret key $sk$.
		\item $\mathsf{Encrypt}(sk, m)\rightarrow ct$: On input a secret key $sk$ and a message
		$m \in D$, the encryption algorithm outputs a ciphertext $ct$.
		\item $\mathsf{Compare}(ct_1, ct_2)\rightarrow b$: On input two ciphertexts $ct_1$, $ct_2$,
		the compare algorithm outputs a bit $b \in \{-1, 0, 1\}$.
	\end{itemize}
\end{definition}\label{dfore}

%As mentioned above, given a ORE encrypted dataset $PE=\{pe_1,pe_2,...,pe_n\}$, the cloud server needs perform comparisons and computations over encrypted data to determine whether $pe_i$ dominates $pe_j$. However, it is no need to calculate the subtraction value, and need to compare the values
%after each subtraction. For example, in plaintext, we only need to compare $T_i[k]=|p_i[k]-q[k]|$ and $T_j[k]=|p_j[k]-q[k]|$ in skyline query, namely the relationship between $|p_i[k]-q[k]|-|p_j[k]-q[k]|$ and $0$. In ciphertext, there are two cases as follow. One is shown in Fig.~\ref{compare1}, $Compare(pe_i[k], qe[k]) \geq 0$ and $Compare(p_j[k], qe[k]) \geq 0$. The relationship between $|p_i[k]-q[k]|-|p_j[k]-q[k]|$ and $0$ can be expressed as $Compare(pe_i[k], pe_j[k])$ and $0$. The other is shown in Fig.~\ref{compare2}, the relationship can be expressed as $Compare(pe_i[k]+pe_j[k], qe\_d[k]+Aqe\_d[k])$ and $0$. All other conditions can be solved in the way similar to that described above.
%
\begin{figure}[t]
	\centering
	\subfloat[Case 1.]{
		\label{compare1}
		\includegraphics[height=2cm]{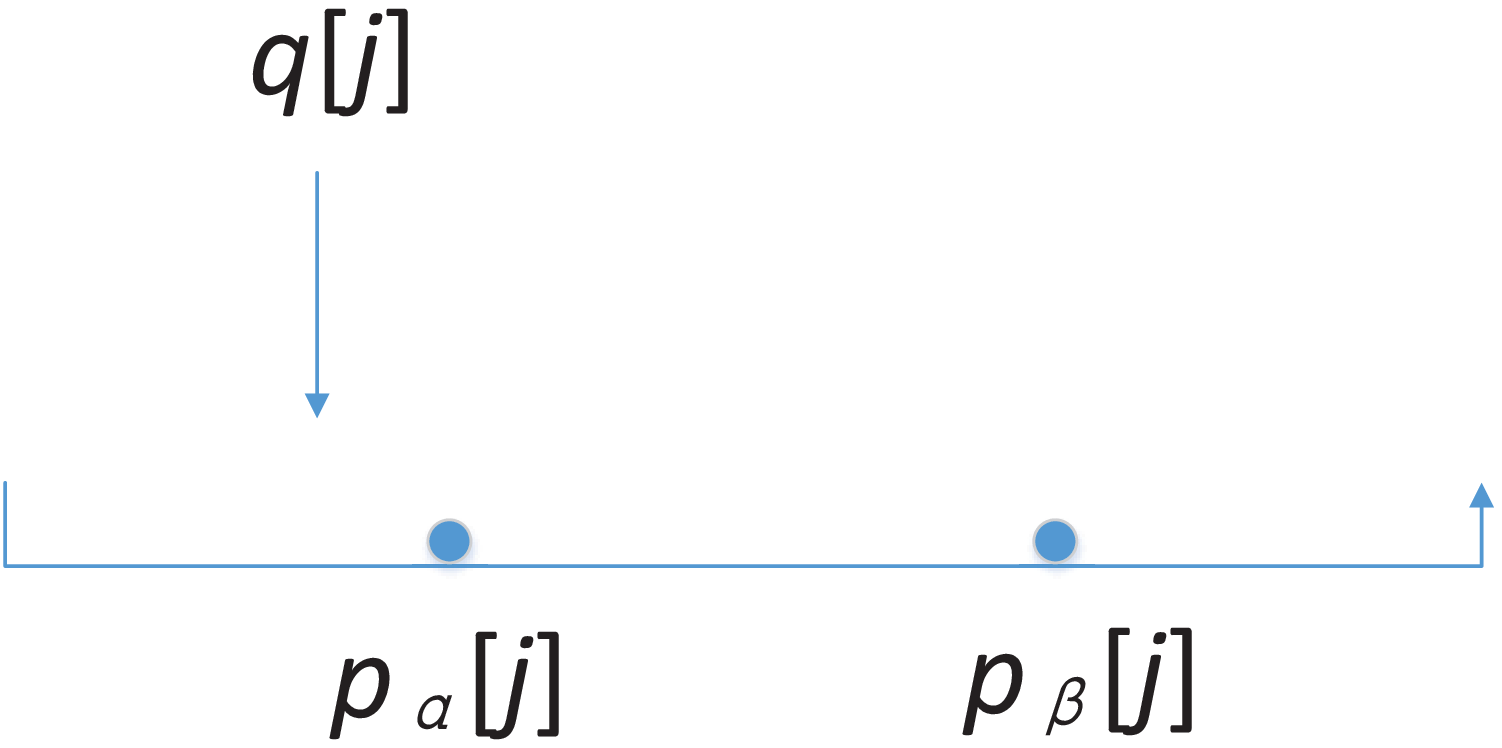}}\quad
	\subfloat[Case 2.]{
		\label{compare2}
		\includegraphics[height=2cm]{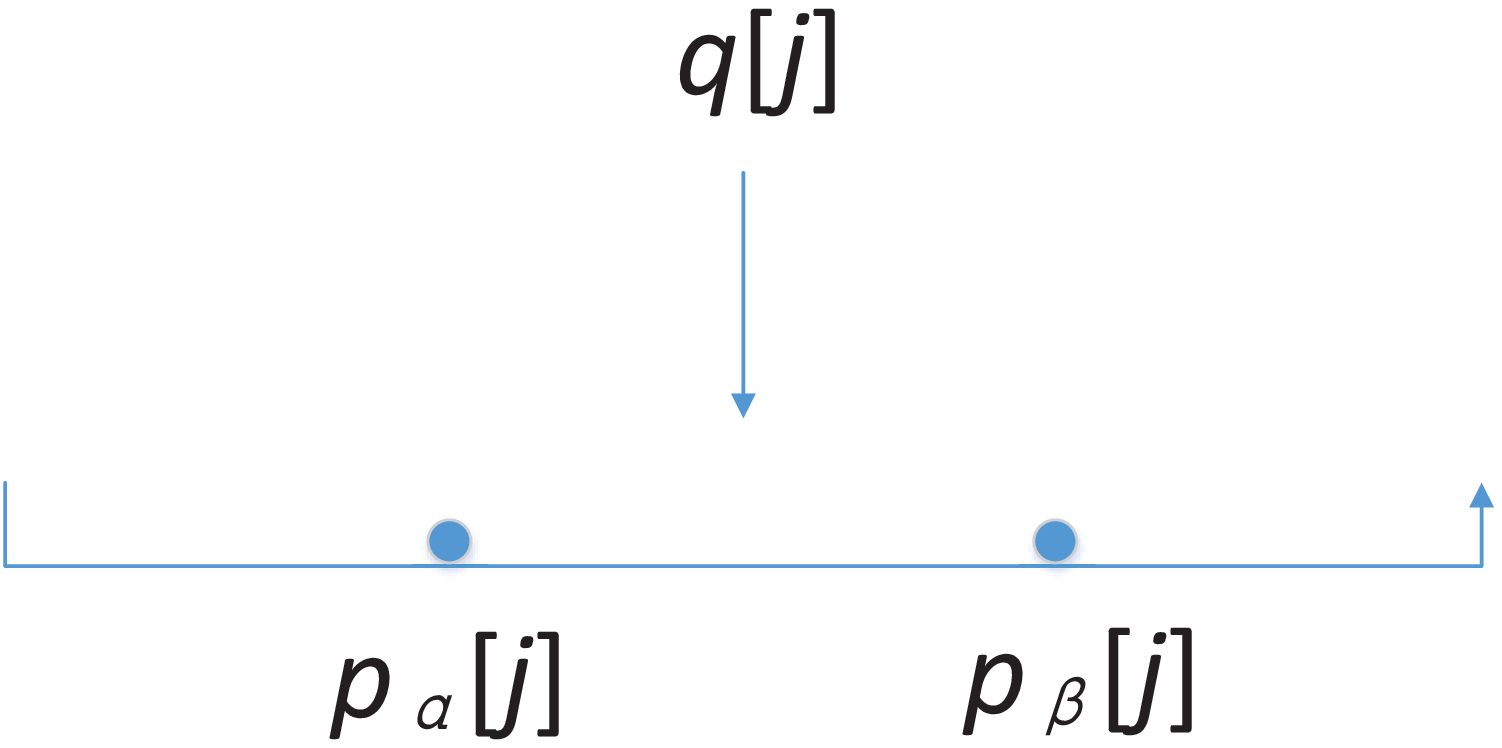}}
	\vspace{0ex}\caption{Cases for the relationship between $q$ and $(p_\alpha,p_\beta)$}
	\label{fig:securecompare}
	\vspace{0ex}\end{figure}

With the help of ORE scheme, evaluating the dynamic domination relation between $p_\alpha$ and $p_\beta$ can be carried out securely in ciphertext form as outlined in Algorithm \ref{securecompare}. For ease of subsequent discussion, we shall denote $\mathsf{Enc}(x)$ as the ORE ciphertext for the original message $x$.

\begin{algorithm}[t]
	\caption{SecureCompare Algorithm}
	\label{securecompare}
	\begin{algorithmic}[1]
		\REQUIRE The ORE ciphertext for $\mathsf{Enc}(p_\alpha[j])$, $\mathsf{Enc}(p_\beta[j])$, $\mathsf{Enc}(q[j])$, as well as $\mathsf{Enc}(p_\alpha[j]+p_\beta[j])$, $\mathsf{Enc}(2q[j])$.
		\ENSURE The comparison result as $-1,0,1$ denoting that $p_\alpha[j]$ is closer to (\resp equivalent with, farther from) $q[j]$ than $p_\beta[j]$.
		
		\IF{$\mathsf{ORE.Compare}(\mathsf{Enc}(p_\alpha[j]),\mathsf{Enc}(p_\beta[j]))==0$}
		\RETURN $0$
		\ELSIF{$\mathsf{ORE.Compare}(\mathsf{Enc}(p_\alpha[j]),\mathsf{Enc}(p_\beta[j]))==-1$}
		\IF{$\mathsf{Enc}(q[j])$ falls outside the interval}
		\RETURN $\mathsf{ORE.Compare}(\mathsf{Enc}(q[j]),\mathsf{Enc}(p_\alpha[j]))$
		\ELSE
		\RETURN $\mathsf{ORE.Compare}(\mathsf{Enc}(2q[j]),\mathsf{Enc}(p_\alpha[j]+p_\beta[j]))$
		\ENDIF
		\ELSE
		\IF{$Enc(q[j])$ falls outside the interval}
		\RETURN $\mathsf{ORE.Compare}(\mathsf{Enc}(q[j]),\mathsf{Enc}(p_\beta[j]))$
		\ELSE
		\RETURN $\mathsf{ORE.Compare}(\mathsf{Enc}(p_\alpha[j]+p_\beta[j]),\mathsf{Enc}(2q[j]))$
		\ENDIF
		\ENDIF
	\end{algorithmic}
\end{algorithm}

\subsubsection{Minimizing the number of keys} 
Following Observation~\ref{obs1}, a data owner needs to encrypt database $P$ and the sum of any two tuples in $P$ in each dimension, namely $p_\alpha[j]+p_\beta[j]$, where $\alpha \neq \beta, \alpha,\beta\in [1,n], j \in [1,d]$. The above two ciphertexts are denoted as $E(P)$ and $E(\Phi)$, respectively. In this step, if we use the same private key on both $E(P)$ and $E(\Phi)$, the sum of paired tuples in $E(\Phi)$, although encrypted, will leak more message about plaintext beyond the order. 

For example, assume that $P$ contains five tuples, whose values in a particular dimension are $a, b, c, d, e$, respectively. Suppose that after sorting the values in ascending order, we get $b, c, a, e, d$. Then their sums can be listed as $b+c, b+a, b+e, b+d, c+a, c+e, c+d, a+e, a+d, e+d$. For ease of discussion, in the following we shall refer to these values as \textit{pairs of sums}. If we encrypt the results for these pairs of sums using the same key as $E(P)$, an attacker can get the ordering of plaintexts. Therefore, he may possibly know $b+e \leq c+a$, and then infer that $e-a \leq c-b$. In this way, besides the order, the distribution of values in plaintext tuples is also leaked. 

However, according to the security model in this work, except the order of tuples in some dimensions, the cloud should not be able to infer the content of the tuples. Therefore, we have to avoid leaking the distribution of
data by adopting different keys in ORE. Intuitively, an ideal method is to encrypt each pair of sums using a different key, as it is not required to perform comparison among any pair of $p_\alpha[j],p_\beta[j]$ according to Algorithm~\ref{securecompare}. However, the increased number of keys will further introduce key management and storage problems. We propose a novel method to address this problem, which is shown in Fig.~\ref{add}.

\begin{figure*}[t]
	\centerline{\includegraphics[width=0.75\textwidth]{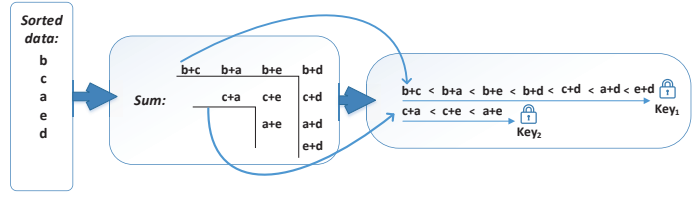}}
	\vspace{0ex}\caption{A novel encryption scheme for pairs of tuples}
	\label{add}
	\vspace{-1ex}\end{figure*}

As shown in Fig.~\ref{add}, $b, c, a, e, d$ are the sorted values for five tuples in $P$ on a particular dimension. According to Algorithm~\ref{securecompare}, these values should be encrypted using the same key as comparisons over their ciphertext are required. As a result, given that $\mathsf{Enc}(b),\ldots,\mathsf{Enc}(d)$ are encrypted using the same key under ORE, any adversary can easily infer that $b+c < b+a < b+e < b+d$ regardless that $b+c,\ldots,b+d$ are encrypted with different keys or not. Therefore, it is not beneficial to use multiple keys for such a group of sums. 
\begin{definition}
	[Order-Obvious Class] Given the order of a set of $n$ elements, whose exact values are unknown, if the order of two summations over paired elements can be inferred, we call them \textit{Order-Obvious}. All the $n(n-1)/2$ paired summations can be divided into several disjoint subsets accordingly, such that all the summations in each subset are \textit{Order-Obvious}. We refer to each subset as an \textit{Order-Obvious Class} (abbrev. \textit{OOC}).
\end{definition}\label{dfooc}

Generally, we can find all \textit{OOC}s, which is classified using lines in Fig.~\ref{add}. The relations for sums in the same \textit{OOC} (\eg line) can be inferred easily purely from $E(P)$. In light of that, we can use the same key to encrypt the sums in the same \textit{OOC}, and adopt different keys across \textit{OOC}s. In this way, any adversary cannot get additional information over the ciphertexts besides the order, and we can effectively minimize the number of keys. In particular, the minimum number of keys, denoted as $\kappa$, (\eg the number of lines in Fig.~\ref{add}) have to satisfy the following theorem.

%\begin{equation}
%number = \lceil\frac{2*n-3}{4}\rceil
%\label{lines}
%\end{equation}

\begin{theorem}\label{thm:keynum}
	In order to satisfy the predefined security model, the minimum number of encryption keys in a dimension should be $\kappa=\lceil\frac{2*n-3}{4}\rceil$.
\end{theorem}

\begin{proof}
	See in Appendix~\ref{apxa}.
\end{proof}

\noindent\textit{Remark.} Through the above strategy, we have minimized the required number of encryption keys. In spite of that, $\kappa$ is still linear to the increase of $n$, which may introduce key management burden if $n$ is very large. To address this, we suggest the following implementations. For each row in Fig.~\ref{add}, we assign it a random $Id_i$. The data owner only needs to store one master key $mk$ and a series of random $Id_i$. Then, $key_i$ for encrypting each row is generated by $mk \oplus Id_i$. In this way, we can effectively generate $\kappa$ different keys based on $mk$.

\subsubsection{Accessing the pairs of sums}
\label{section:AccessingPair}
As required by Algorithm~\ref{securecompare}, in order to compare $t_\alpha[j]$ and $t_\beta[j]$, it is always required to retrieve the ciphertext of $p_\alpha[j]+p_\beta[j]$. Therefore, it is necessary to build a map between the elements of $E(P)$ with the corresponding sums in $E(\Phi)$. That is, we need to build an index that maps $\mathsf{Enc}(p_\alpha[j])$ and $\mathsf{Enc}(p_\beta[j])$ to $\mathsf{Enc}(p_\alpha[j]+p_\beta[j])$. To this end, we present an index based on hash function. Formally, we define a hash function as $h:\mathbb{N}^2\rightarrow\mathbb{N}$, where $\mathbb{N}$ denote the set of natural numbers. The hash function $h$ should satisfy the following property, $\forall x_1,y_1,x_2,y_2\in\mathbb{N}$, $h(x_1,y_1)=h(x_2,y_2)$ if and only if $x_1=x_2$ and $y_1=y_2$.

Assume the indices for $\mathsf{Enc}(p_\alpha[j])$ and $\mathsf{Enc}(p_\beta[j])$ in $E(P)$ are denoted as $\alpha$ and $\beta$, respectively. Then the index of $\mathsf{Enc}(p_\alpha[j]+p_\beta[j])$ in $E(\Phi)$ can be easily acquired as $h(\alpha,\beta)$.
Fig.~\ref{Indexes} presents an example for the hash function. There are five encrypted values in $E(P)$, namely $a,\ldots,e$. The hash function in this example is simply designed as a regular traversal order for the corresponding sums. In fact, any hash function that satisfies the aforementioned property can be adopted here.

\subsubsection{Indexing the pairs of sums}
Additionally, as all the pairs of sums within a particular \textit{OOC} are encrypted by ORE using the same key, we present to use additional index structures for efficient retrieval of corresponding entries for these \textit{pairs of sums}. Therefore, we also design an index scheme for management of these ORE encrypted \textit{pairs of sums}. In \textsc{scale}, we adopt AVL-Tree based structure to construct the indexing structure, as it provides the best efficiency when querying for a particular range. Specifically, it is possible for us to build an AVL-Tree to index all these encrypted sums in the same \textit{OOC}. Notably, each AVL-Tree roots at the median of each \textit{OOC} and all the nodes in an AVL-Tree are the corresponding ciphertexts for \textit{pairs of sums} in the same \textit{OOC}. 

For instance, given the records in Fig.~\ref{add}, there are two \textit{OOC}s. We shall build two different AVL-Trees for indexing the corresponding ciphertexts for each \textit{OOC}, respectively. That is, the first \textit{OOC} centered at $b+d$ corresponds to an AVL-Tree rooted with $\mathsf{Enc}(b+d)$; another \textit{OOC} centered at $c+e$ corresponds to another AVL-Tree rooted with $\mathsf{Enc}(c+e)$ (as shown in Fig.~\ref{Indexes}). 

In fact, data structures other than AVL-Tree can also be adopted to index the ORE ciphertexts for each \textit{OOC}. We select AVL-Tree as the default setting in \textsc{scale} as it provides the best query response time among all the choices. Detailed comparison between AVL-Tree and other choices will be discussed in Section~\ref{ssec:43}.

\begin{figure}[t]
	\centerline{\includegraphics[width=0.8\linewidth,height=3.6cm]{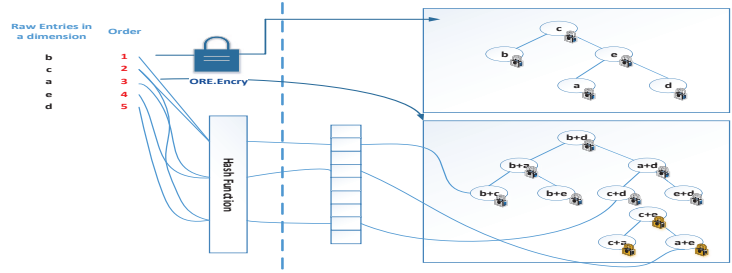}}
	\vspace{-2ex}\caption{The complete ciphertext storage structure}
	\label{Indexes}
	\vspace{0ex}\end{figure}

\subsubsection{Database encryption}\label{sssec:dbenc}
We have now all the ammunitions in place to demonstrate the entire process of encrypting the database (Algorithm \ref{encryptdatasets}). First, the data owner generates $d+\lceil\frac{2*n-3}{4}\rceil*d$ keys (Line 1), and for each column (\ie attribute) in $P$ we encrypt the entries using the same key (Lines 3-5), resulting in $E(P)$. Then, the data owner sorts the entries (Line 8) in each column (\ie attribute) and computes the sums for pairs of entries in each dimension. Afterwards, the sums are then encrypted using the corresponding keys as shown in Fig.~\ref{add} (Lines 9-17), resulting in $E(\Phi)$. Finally, the data owner sends $E(P)$, $E(\Phi)$ to the cloud server.

\begin{algorithm}[t]
	\caption{Dataset Encryption}
	\label{encryptdatasets}
	\begin{algorithmic}[1]
		\REQUIRE The dataset $P$
		\ENSURE The ciphertexts sets $E(P)$, $E(\Phi)$
		\STATE generate $d+\lceil\frac{2*n-3}{4}\rceil*d$ keys with $\mathsf{ORE.Setup}$ as $keys[]$
		\FOR{$p\in P$ and $j$ in $1,\ldots,d$}
		\STATE $\mathsf{Enc}(p[j])\leftarrow\mathsf{ORE.Encrypt}(keys[j],p[j])$
		\STATE let $\mathsf{Enc}(p)=\{\mathsf{Enc}(p[1]),\ldots,\mathsf{Enc}(p[d])\}$ and add $\mathsf{Enc}(p)$ to $E(P)$
		\ENDFOR
		%\STATE $index \Leftarrow Index(n)$
		%\STATE $add\_num \Leftarrow \lceil\frac{2*n-3}{4}\rceil$
		\STATE let $m=1$
		\FOR{$j$ in $1,\ldots,d$}
		\STATE $\Lambda=(p^{(1)}[j],\ldots,p^{(n)}[j])\leftarrow$ sort $p_1[j],\ldots,p_n[j]$ in ascending order
		\WHILE{$\Lambda$ is not empty}
		\FOR{$i$ in $2,\ldots,len(\Lambda)$}
		\STATE add $\mathsf{ORE.Encrypt}(keys[d+m],p^{(1)}[j]+p^{(i)}[j])$ to $E(\Phi)$
		\ENDFOR
		\FOR{$i$ in $2,\ldots,len(\Lambda)-1$}
		\STATE add $\mathsf{ORE.Encrypt}(keys[d+m],p^{(n)}[j]+p^{(i)}[j])$ to $E(\Phi)$
		\ENDFOR
		\STATE remove the first and last elements in $\Lambda$, let $m=m+1$
		\ENDWHILE
		\ENDFOR
		\RETURN $E(P)$, $E(\Phi)$
	\end{algorithmic}
\end{algorithm}

Besides, the data owner also creates a hash function $h$ that maps each pair of elements in $E(P)$ and the corresponding sums in $E(\Phi)$, and sends $h$ to the cloud server. It is now possible for the cloud to quickly locate the ciphertext of the corresponding sums for each pair $(p_\alpha[j], p_\beta[j])$.

\vspace{-1ex}\subsection{Query Processing}

A data user needs to register their information to the Data owner and securely get the keys. Then the query user encrypts the request according to Algorithm \ref{encryptquery} before sending it to the cloud server.

\begin{algorithm}[t]
	\caption{Query Request Encryption}
	\label{encryptquery}
	\begin{algorithmic}[1]
		\REQUIRE The query data $q$, keys from data owner $keys[]$
		\ENSURE The ciphertexts $\mathsf{Enc}(q)$, $\mathsf{Enc}(2q)$
		\FOR{$j$ in $1,\ldots,d$ }
		\STATE $\mathsf{Enc}(q[j])\leftarrow \mathsf{ORE.Encrypt}(keys[j],q[j])$
		\FOR{$m$ in $1,\ldots,\lceil\frac{2*n-3}{4}\rceil$}
		\STATE let $key\_num=d+(j-1)*\lceil\frac{2*n-3}{4}\rceil+m$
		\STATE $\mathsf{Enc}(2q_m[j])\leftarrow\mathsf{ORE.Encrypt}(keys[key\_num],2q[j])$
		\ENDFOR
		\ENDFOR
		\RETURN $\mathsf{Enc}(q)$, $\mathsf{Enc}(2q)$
	\end{algorithmic}
\end{algorithm}

As shown in Algorithm \ref{encryptquery}, user encrypts each dimension of the query tuple using corresponding keys (Line 2) and encrypts the doubled entries for the query tuple using other keys (Lines 4-5). Finally, the user sends $\mathsf{Enc}(q)$, $\mathsf{Enc}(2q)$ to the cloud server. As mentioned in Algorithm \ref{skylinequery}, given an encrypted query $q$ as shown Algorithm~\ref{encryptquery}, the cloud server needs to perform comparisons and computations over encrypted data. According to the approach shown in Fig.~\ref{fig:securecompare}, the cloud server can perform skyline query via the comparison relationship with encrypted tuples, encrypted query request, encrypted sums, and encrypted doubled request. As a result, the process described in Algorithm \ref{skylinequery} can be now performed in ciphertext without decryption, which is shown in Algorithm \ref{secureskylinequery}.

\begin{algorithm}[t]
	\caption{Secure Skyline Query Algorithm}
	\label{secureskylinequery}
	\begin{algorithmic}[1]
		\REQUIRE The ciphertext for dataset $E(P)$, query request $\mathsf{Enc}(q)$, sums for tuples $E(\Phi)$ and doubled query request $\mathsf{Enc}(2q)$
		\ENSURE The encrypted result set of skyline points $S$
		\FOR{$i$ in $1,\ldots,n$}
		\IF{$S$ is empty}
		\STATE add $\mathsf{Enc}(p_i)$ to $S$
		\ELSE
		\STATE $flag\_cur \leftarrow True$
		\FOR{each $\mathsf{Enc}(p_j)\in S$}
		\FOR{$m$ in $1,\ldots,d$}
		\STATE $flag[m]\leftarrow$\textit{SecureCompare}($\mathsf{Enc}(p_i[m])$, $\mathsf{Enc}(p_j[m])$, $\mathsf{Enc}(q[m])$, $\mathsf{Enc}(p_i[m]+p_j[m])$, $\mathsf{Enc}(2q[m])$)
		%				\STATE Compute the order $oi=Order(m,i)$ and $oj=Order(m,j)$
		%				\STATE Compute the $d1=Min(oi, oj)$ and $d2=Max(oi, oj)$
		%				\STATE Compute the sum of $i$ and $j$, get $pe\_d_{Index(d1,d2)}$
		%				\STATE Securecompare($pe_i[m]$, $pe_j[m]$, $pe\_d_{Index(d1,d2)}$, $qe[m]$, $qe\_d[m]$)
		\ENDFOR
		\IF{$\forall m, flag[m]\ge 0$, and $\exists k$ such that $flag[k]>0$}
		\STATE $flag\_cur\leftarrow False$
		\ELSIF{$\forall m, flag[m]\le 0$, and $\exists k$ such that $flag[k]<0$}
		\STATE delete $\mathsf{Enc}(p_j)$ from $S$
		\ENDIF
		\ENDFOR
		%		\IF{no one in $S$ dominates $T_i$ in each dimension}
		\IF{$flag\_cur$ is $True$}
		\STATE add $\mathsf{Enc}(p_i)$ to $S$
		\ENDIF
		\ENDIF
		\ENDFOR
		\RETURN $S$
	\end{algorithmic}
\end{algorithm}

To illustrate the entire protocol, we provide a running example in the following. 
\begin{example}
	For the convenience of representation, we assume that $P$ contains five tuples, whose entries in dimension $1$ are sorted as $7, 13, 21, 32, 53$. 
	
	According to Algorithm~\ref{encryptdatasets}, we shall first compute the sums for all pairs of values, \eg $7+13=20$, $7+21=28$, $7+32=39$, $7+53=60$, $13+21=34$, $13+32=45$, $13+53=66$, $21+32=53$, $21+53=74$, $32+53=85$. As shown in Theorem~\ref{thm:keynum}, the number of encryption keys required for these sums can be calculated as $\lceil\frac{2*5-3}{4}\rceil=2$. Therefore, we use two keys to encrypt the above sums, resulting in $\mathsf{Enc}_1(20)$, $\mathsf{Enc}_1(28)$, \ldots, $\mathsf{Enc}_1(85)$ and $\mathsf{Enc}_2(34)$,$\mathsf{Enc}_2(45)$,$\mathsf{Enc}_2(53)$. 
	
	Besides, we also need to use another key to encrypt the original tuples, \eg $\mathsf{Enc}_3(7), \\\mathsf{Enc}_3(13), \ldots, \mathsf{Enc}_3(53)$. Suppose that a user submits a query with $q[1]=23$. Then $q$ and $2q$ need to be encrypted according to our scheme, resulting in $\mathsf{Enc}_1(46), \mathsf{Enc}_2(46), \linebreak\mathsf{Enc}_3(23)$. These ciphertexts are then sent to the cloud server. 
	The cloud server compares ciphertexts one by one according to the protocol. Through $\mathsf{ORE.Compare}$ and Algorithm~\ref{securecompare}, the cloud server can easily determine that $\mathsf{Enc}_3(32)$ dominates $\mathsf{Enc}_3(53)$ following the case shown in Fig.~\ref{compare1}. Similarly, $\mathsf{Enc}_3(21)$ dominates $\mathsf{Enc}_3(7)$ and $\mathsf{Enc}_3(13)$. In the case shown in Fig.~\ref{compare2}, $\mathsf{Enc}_3(21)$ dominates $\mathsf{Enc}_3(32)$ because $\mathsf{ORE.Compare}(\mathsf{Enc}_2(\\53),\mathsf{Enc}_2(46))=1$. Algorithm \ref{secureskylinequery} will iteratively repeat this process for all dimensions and remaining tuples.\EndOfProof
\end{example}

\vspace{-2ex}\subsection{Maintenance Issues}\label{ssec:43}
Modifications over database records (\textit{insert}, \textit{delete}, \textit{update}) are fundamental requirements in database applications. In light of that, hereby we discuss the strategies to support these operations in our framework. As depicted in Fig.~\ref{add}, the cloud server stores these encrypted sums of data values in different \textit{OOC}s, which contributes the most expensive maintenance cost. Hence, the way how these encrypted sums are stored is fundamentally important. In fact, many index structures can be used to accomplish this task. In \textsc{scale}, we adopt AVL-Tree \cite{Adel1962An}, as it presents the best efficiency in searching for an entry due to the strictly balanced structure. In fact, we have considered and compared several different structures including Linked list, AVL-Tree and Red-black tree. Table \ref{tb:advAndDisadv} shows the functional comparison over the advantages and disadvantages of these data structures in our models. In the following, we shall sequentially describe how insertion and deletion is supported in \textsc{scale} using AVL-Tree.

\begin{table}[t]
	\caption{Functional comparison over different structures for insertion and deletion}\vspace{-1ex}
	\small
	\begin{center}
		\begin{tabular}{|p{1.6cm}|p{2.5cm}|p{3.4cm}|}
			\hline
			\textbf{Data Structure} & \textbf{Advantages} & \textbf{Disadvantages} \\
			\hline
			Linked List &  Easy implementation & Expensive cost \\
			\hline
			AVL-Tree & Shorter query time than Red-black Tree & Longer response time for the insert and delete operations \\
			\hline
			Red-black Tree &  Longer query time than AVL-Tree & Shorter response time for the insert and delete operations \\
			\hline
		\end{tabular}\label{tb:advAndDisadv}
	\end{center}\vspace{-1ex}
\end{table}

\subsubsection{Insertion} 
As described in Section~\ref{ssec:41}, the entries for each records should be encrypted in multiple copies. Therefore, any newly inserted records have also to undertake the same procedure. For example, assume that the data owner adds a new tuple $f$ to the existing database described in Fig.~\ref{Indexes}, which contains $a,b,c,d,e$, and $b<c<f<a<e<d$. In the manner described above, the data owner encrypts the tuple for each dimension using different keys. Afterwards, the data owner computes $b+f$, $c+f$, $f+a$, $f+e$, $f+d$ and encrypts them with different groups (Fig.~\ref{insert}) of keys. Finally, these encrypted values are uploaded to the server. 

\begin{figure}[t]
	\centerline{\includegraphics[width=0.8\columnwidth]{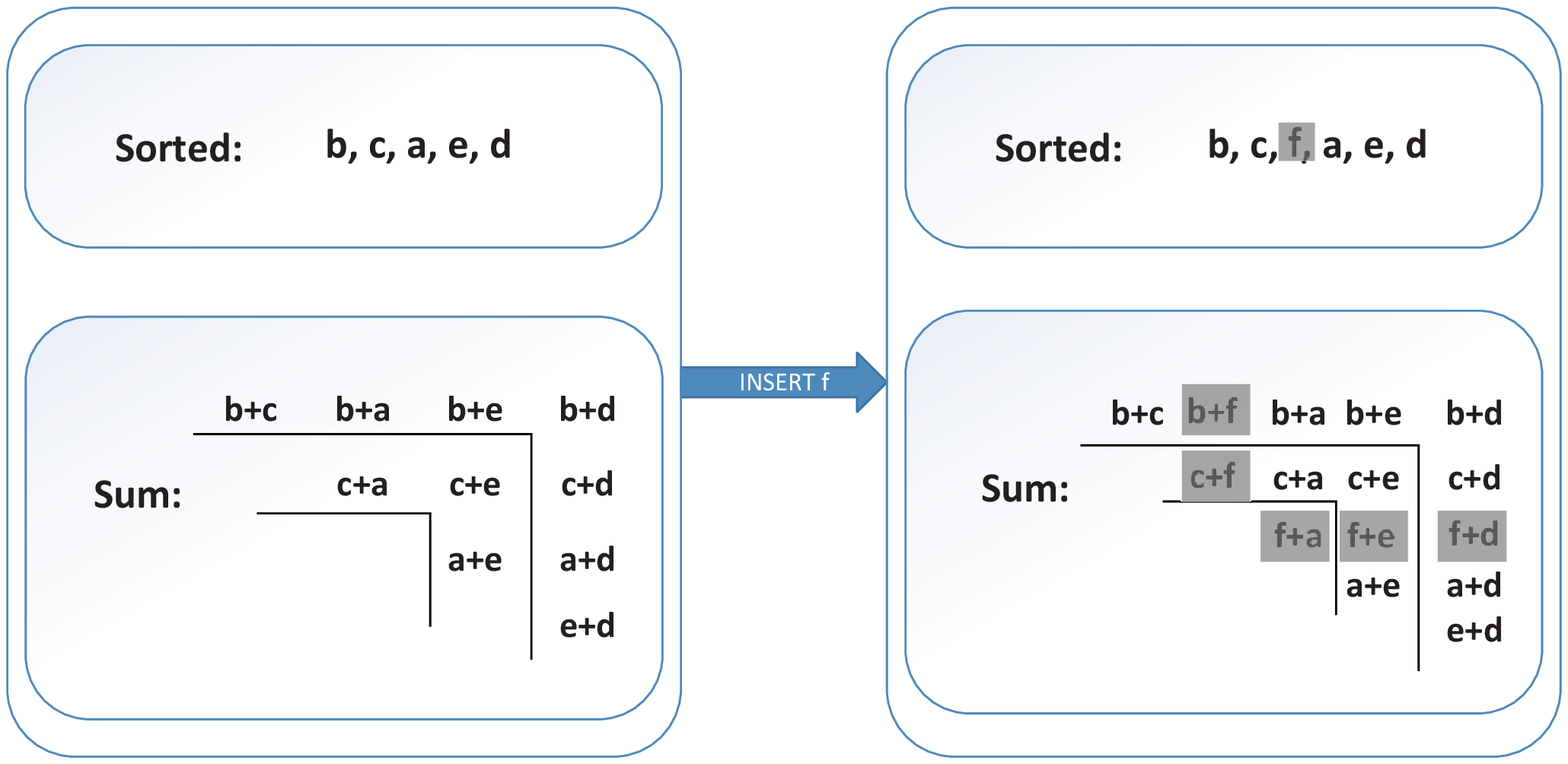}}
	\vspace{-1ex}\caption{Insert a new tuple $f$}
	\label{insert}
	\vspace{-1ex}\end{figure}

As depicted in Fig.~\ref{avl_insert_1}, the cloud server adopt an $AVL-tree$ to store encrypted sums of entries for each \textit{OOC} shown in Fig.~\ref{insert}. For instance, there are two \textit{OOC}s in the left part of Fig.~\ref{insert}, which show the \textit{OOC} for existing records, we implement two corresponding AVL-Trees to store the \textit{OOC}s, \eg one tree rooted at $b+d$ and contains all the sums in the same \textit{OOC}, and the other tree rooted at $c+e$ and contains two other sums, namely $c+a$ and $a+e$. As shown in Fig.~\ref{insert} and~\ref{avl_insert_2}, $b+f$, $f+d$ belongs to the same \textit{OOC} with $b+d$ and can be inserted into the corresponding positions in the AVL-Tree rooted with $b+d$. Similarly, the corresponding ciphertexts for $c+f$, $f+e$ will be inserted into another AVL-Tree rooted with $c+e$. Moreover, $f+a$ will be inserted to a new AVL-Tree rooted with $f+a$.

\begin{figure}[t]
	\centering
	\subfloat[Before insertion.]{
		\label{avl_insert_1}
		\includegraphics[width=0.45\columnwidth]{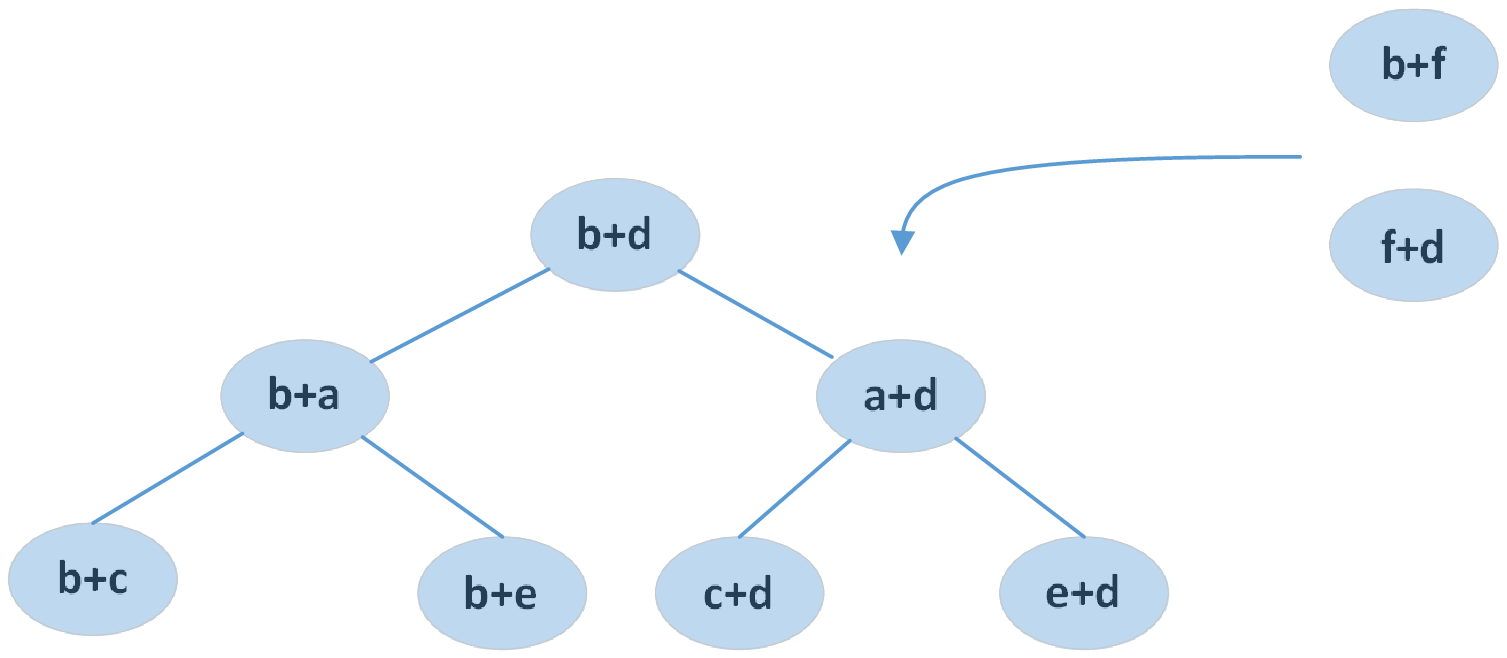}}
	\quad
	\subfloat[After insertion.]{
		\label{avl_insert_2}
		\includegraphics[width=0.45\columnwidth]{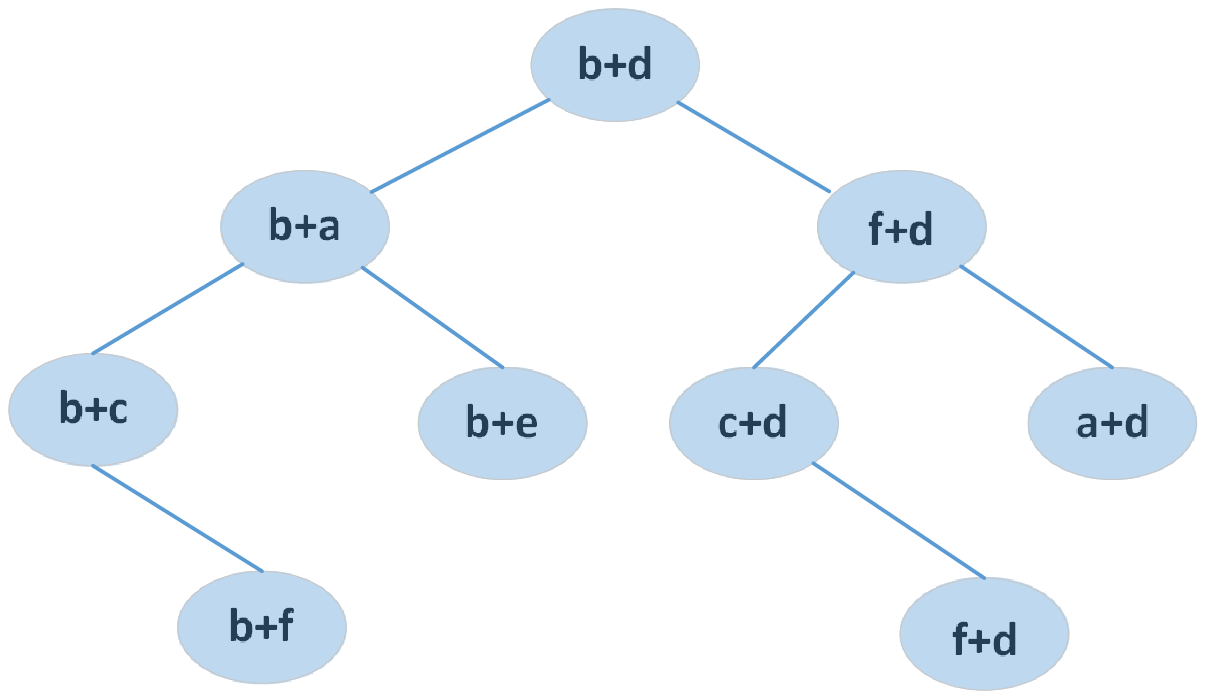}}
	\vspace{-1ex}\caption{Inserting $b+f$ and $f+d$ into the AVL-Tree}
	\label{fig:insertTuple}
	\vspace{-1ex}\end{figure}

\subsubsection{Deletion} 
On the other hand, data owners also may have to delete tuples from existing database. In this scenario, \textsc{scale} also need to update the corresponding indices that are associated with the deleted records. For example, assume that the data owner want to delete a tuple $e$ in $b, c, a, e, d$. In the manner described in Section~\ref{ssec:41}, data owners encrypt the tuple for each dimension using different keys. Besides, the data owner have already computed the sums including $b+e$, $c+e$, $a+e$, $d+e$ and encrypt them with different groups (Fig.~\ref{delete}) of keys. All the corresponding ciphertexts have been uploaded and stored in the cloud server. 

\begin{figure}[t]
	\centerline{\includegraphics[width=0.78\columnwidth]{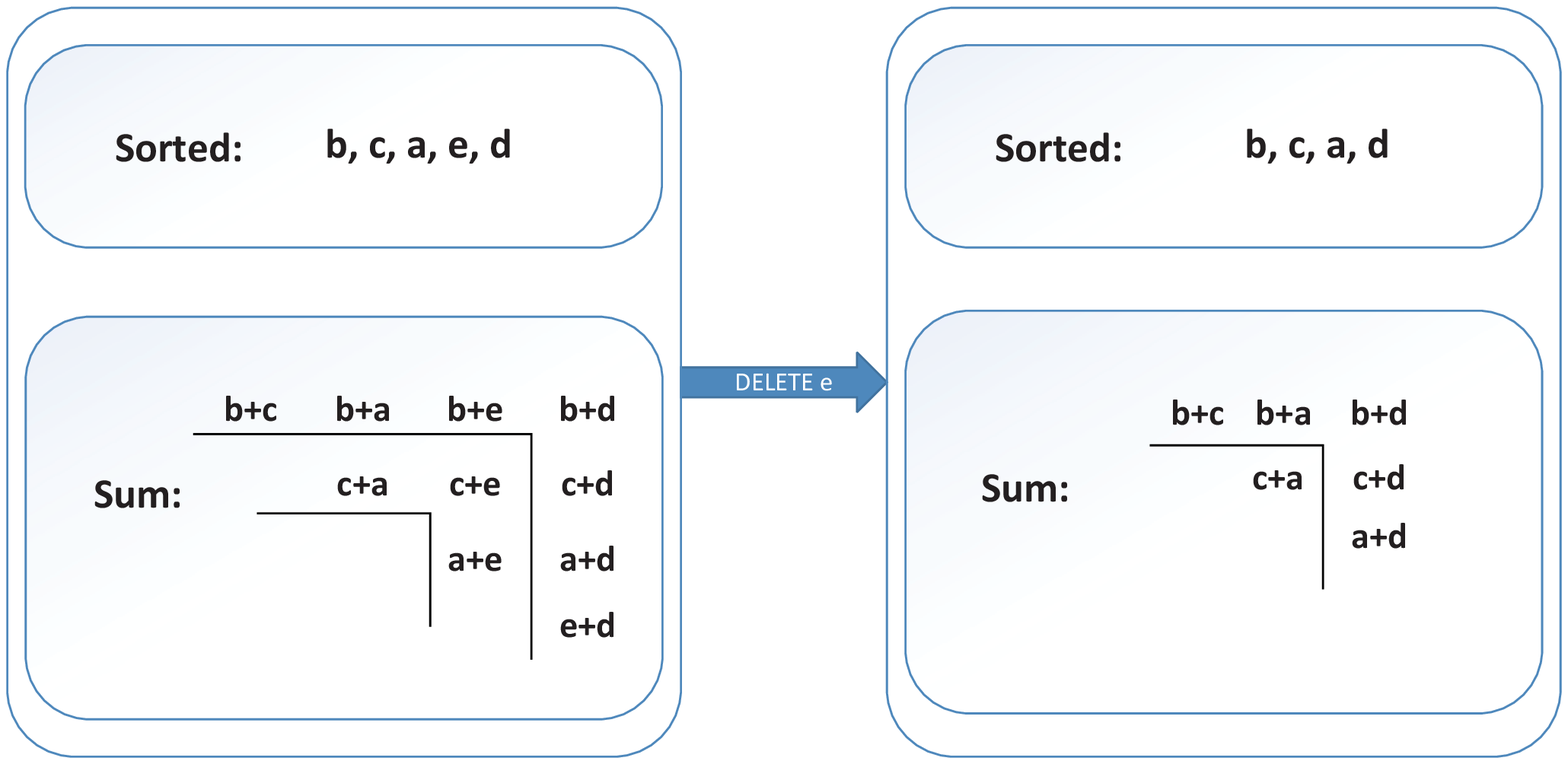}}
	\vspace{-1ex}\caption{Delete a tuple $e$}
	\label{delete}
	\vspace{-1ex}\end{figure}

As depicted in Fig.~\ref{avl_delete_1}, in \textsc{scale}, the cloud server use an $AVL-tree$ to store the sums in each \textit{OOC} with the same key. AVL-Tree structure provides efficient deletion efficiency. Whenever $e$ is deleted from the database, the corresponding sums with respect to $e$, \eg $b+e$, $e+d$, $c+e$ and $a+e$, shall also be removed from the corresponding AVL-Trees. In particular, $b+e$ and $e+d$ shall be removed from the AVL-Tree rooted at $b+d$ (as shown in Fig.~\ref{avl_delete_2}), which will be then balanced accordingly afterwards; similarly, $c+e$, $a+e$ will be removed from another AVL-Tree rooted at $c+e$. 

\subsubsection{Update}
Notably, all the existing data records are stored in ciphertext form according to our framework. In ciphertext space, the \textit{update} operation cannot be directly applied. Instead, it is interpreted as deleting an existing encrypted record and then insert a new encrypted record. 

\begin{figure}[t]
	\centering
	\subfloat[Before deletion.]{
		\label{avl_delete_1}
		\includegraphics[width=0.45\columnwidth]{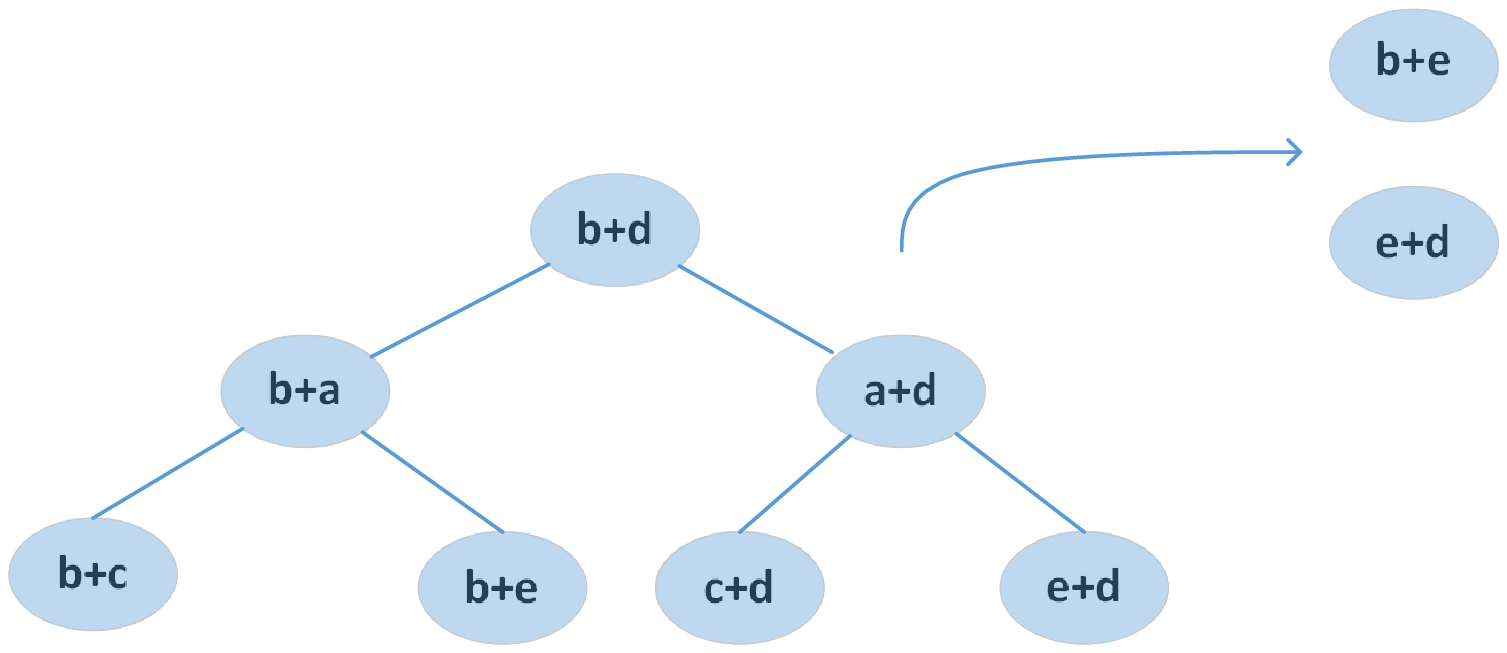}}
	\quad
	\subfloat[After deletion.]{
		\label{avl_delete_2}
		\includegraphics[width=0.45\columnwidth]{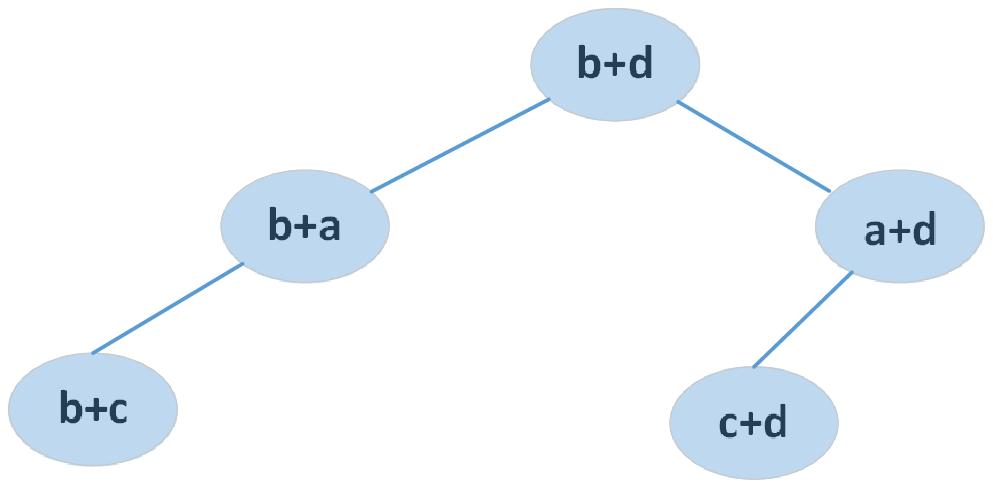}}
	\vspace{-1ex}\caption{Deleting $b+e$ and $e+d$ from the AVL-Tree}
	\label{fig:deleteTuple}
	\vspace{-1ex}\end{figure}

%\vspace{0ex}\subsection{Correction Analysis}\label{ssec:correct}
%\btitle{Query.} The data owner first sends the two ciphertexts sets $E(P)$, $E(\Phi)$ to the server. The data user, in order to issue a query in single setting, sends $\mathsf{Enc}(q)$, $\mathsf{Enc}(2q)$ to the server. As depicted in Algorithm \ref{secureskylinequery}, The server can compare the ciphertext data according to the method of Observation 1 correctly. In the case of Fig. \ref{compare1}, a comparison between $E(P)$ and $\mathsf{Enc}(q)$ is performed directly. In the case of Fig. \ref{compare2}, the data of different groups as shown in Fig. \ref{add} are compared using the encrypted query data $\mathsf{Enc}(2q)$ and the ciphertext sets $E(\Phi)$ by different keys, and the skyline set to be obtained can be calculated.
%
%\btitle{Deletion And Update.} The data owner, in order to delete an old tuple or update a new tuple, send the encrypted tuple and the ciphertext of the tuple and other tuple by the way as shown in Fig. \ref{add}. The server updates or deletes the data corresponding to the uploaded data. The above process does not change the information of the remaining tuple, so it does not change the correctness of the query results.

\vspace{0ex}\subsection{Security Analysis}\label{ssec:secana}
Insert and delete operations on AVL-Tree may provide opportunities for side channel attacks. But it is not the key point of this paper and some existing work can solve this problem. It would not be discussed here. The presented \textsc{scale} framework is constructed based on ORE scheme proposed in~\cite{Lewi2016Order}, which is secure with leakage function $\mathcal{L}_{BLK}$. The particular lemma is defined as follows.
\begin{lemma}
	The ORE scheme is secure with leakage function $\mathcal{L}_{BLK}$ assuming that the adopted \emph{pseudo random function} (PRF) is secure and the adopted hash functions are modeled as random oracles. Here, 
	$\mathcal{L}_{BLK}(m_1,\ldots,m_t)=\{(i,j,BLK(m_i,m_j))|1 \leq i < j \leq t\}$ and $BLK(m_i, m_j)=(ORE.Compare(m_i, m_j),\mathtt{ind}_{\mathtt{diff}}(m_i,m_j))$, in which $\mathtt{ind}_{\mathtt{diff}}$ is the first differing block function that is the first index $i\in[n]$ such that $x_i=x_j$ for all $j<i$ and $x_i \neq x_j$. (The proof of this lemma is in Appendix 4.1 in~\cite{Lewi2016Order} and is omitted here.)
\end{lemma}
In order to formally prove the security of \textsc{scale}, we extend $\mathcal{L}$-adaptively-secure model for keyword searching scheme as shown in Definition~\ref{df:AdaptivelySecure}.
\begin{theorem}
	\label{theorem2}
	Let the adopted PRF in ORE is secure. The presented \textsc{scale} framework is $\mathcal{L}$-adaptively-secure in the (programmable) random oracle model, where the leakage function collection $\mathcal{L} = (\mathcal{L}_{Encrypt}, \mathcal{L}_{Query}, \mathcal{L}_{Insert}, \mathcal{L}_{Delete})$ is defined as follows,
	\begin{align*}
		\footnotesize
		\mathcal{L}_{Encrypt}= \mathcal{L}_{BLK}(\cup_{k=1}^{d}X_k^{(n)}),	\mathcal{L}_{Query}= \mathcal{L}_{BLK}(\cup_{k=1}^{d}X_k^{(n)'}),\\
		\mathcal{L}_{Insert}= \mathcal{L}_{BLK}(\cup_{k=1}^{d}X_k^{(n+1)}),
		\mathcal{L}_{Delete}= \mathcal{L}_{BLK}(\cup_{k=1}^{d}X_k^{(n)})
	\end{align*}
	where $X_k^{(n)}=\cup_{t=1}^{(\lceil\frac{2*n-3}{4}\rceil)}(Y_t^{k})$, $X_{k}^{(n)'}=\cup_{t=1}^{(\lceil\frac{2*n-3}{4}\rceil)}(Y_t^{k}\cup \mathsf{Enc}_t(q) \cup \mathsf{Enc}_t(2q))$ and $Y_t^{k}=\{\mathsf{Enc}(p_t[k]+p_j[k])|t<j<n-t+1\}\cup\{\mathsf{Enc}(p_j[k]+p_{(n-t+1)}[k])|t<j<n-t+1\}$.
	
	The advantage for any probabilistic polynomial-time adversary is,
	\begin{equation*}
		\begin{split}
			|{\rm Pr}&[{\rm Game}_{\mathcal{R},\mathcal{A}(\lambda)}=1]-{\rm Pr}[{\rm Game}_{\mathcal{S}, \mathcal{A}(\lambda)} = 1]|\\
			&\leq {\rm negl}(\lambda)=d\cdot({\rm negl}^{ORE}(\lambda) + (2n-1){\rm poly}(\lambda)/2^{\lambda}).
		\end{split}
	\end{equation*}
\end{theorem}

\begin{proof}See in Appendix~\ref{apxb}.
\end{proof}

\vspace{-1ex}\subsection{Complexity Analysis}\label{ssec:compcompana}

In the encryption phase, the plaintext data from the data owner can be sorted and encrypted in advance. We need $O(d+\lceil\frac{2*n-3}{4}\rceil)$ encryption operations every time when a user submits a query following Algorithm \ref{encryptquery}.

During the querying phase, our scheme replaces the original plaintext subtraction and comparison operations with a limited number of comparisons over ciphertext. The time taken for encryption and ciphertext comparisons is only affected by the block size in ORE, key length in AES, and plaintext length. Therefore, we have not changed the main logic for dynamic skyline query processing. Hence, the complexity for the query processing phase in our scheme is consistent with that of~\cite{borzsony2001skyline}. That is, the complexity is $O(n^2)$ for the worst case.

Afterwards, we shall discuss the complexity for inserting and deleting operations. In particular, the main complexity in insertion and deletion of records lie in the update of corresponding AVL-Trees for the sums of records pairs. Notably, the number of AVL-Trees to store these sums should be equal to the number of \textit{OOC}s. As described above, the number of \textit{OOC}s is $\lceil\frac{2*n-3}{4}\rceil$. The time complexity of inserting and deleting elements should be $\lceil\frac{2*n-3}{4}\rceil$ times the complexity of inserting and deleting elements in different data structures. It's important to note that we need to find the corresponding positions before deleting and inserting elements. Table \ref{tb:insertAndDeleteTC} shows the time complexity for inserting and deleting elements by adopting different data structures other than AVL-Tree. 

\begin{table}[t]
	\small
	\caption{The time complexities of inserting and deleting a record using linked list, AVL-Tree, and Red-black Tree}\vspace{-1ex}
	\begin{center}
		\begin{tabular}{|p{2.6cm}|p{2cm}|p{2cm}|}
			\hline
			\textbf{Data Structure} & \textbf{Insertion} & \textbf{Deletion} \\
			\hline
			Linked list &  $O(n^2)$ & $O(n^2)$ \\
			\hline
			AVL-Tree & $O(nlogn)$ & $O(nlogn)$ \\
			\hline
			Red-Black tree &  $O(nlogn)$  & $O(nlogn)$ \\
			
			\hline
		\end{tabular}\label{tb:insertAndDeleteTC}
	\end{center}\vspace{-1ex}
\end{table}

\begin{figure*}[t]
	\centering
	\subfloat[CORR] { \label{impactn:a}
		\includegraphics[width=0.24\columnwidth]{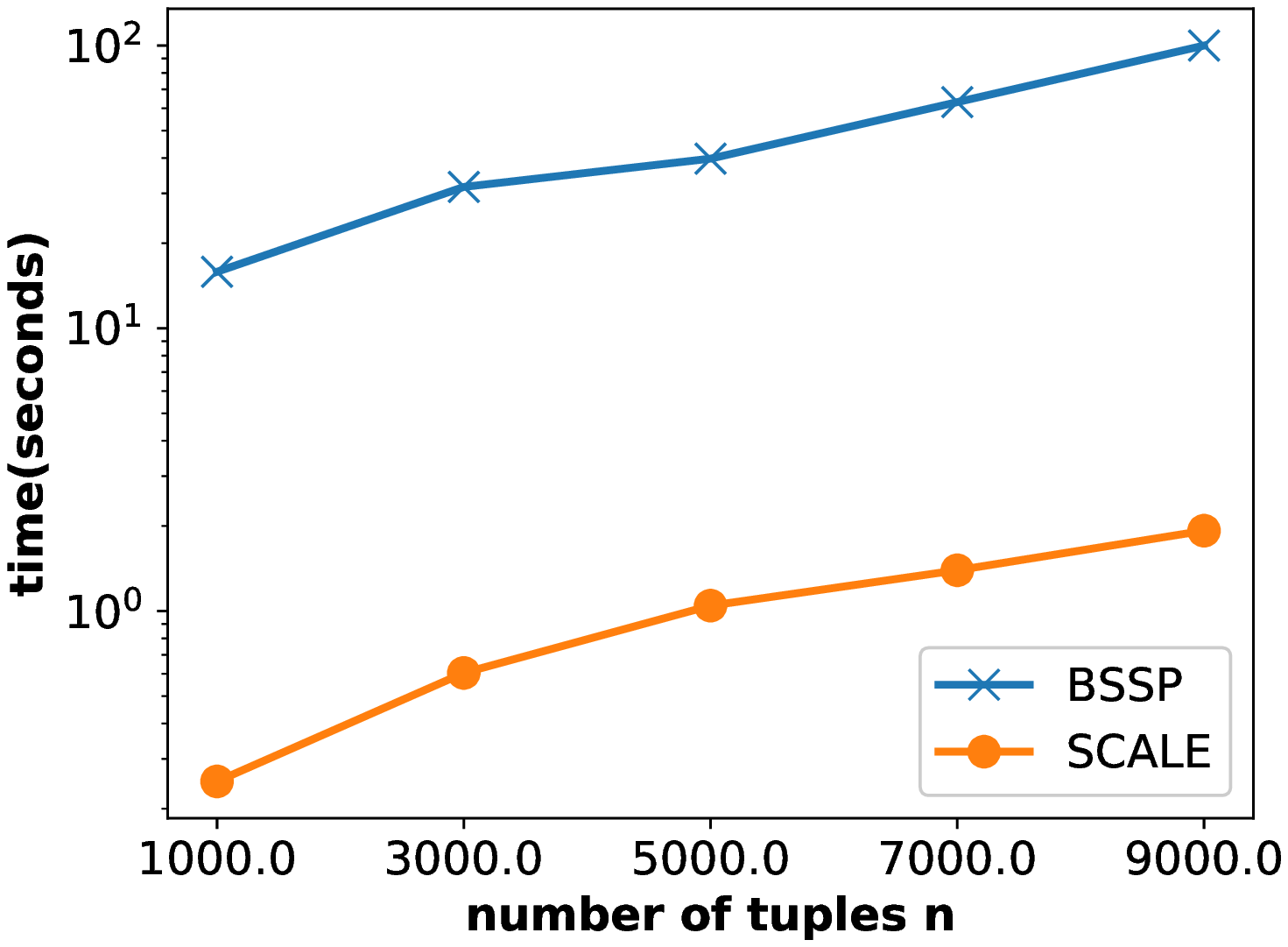}
	}
	\subfloat[ANTI] { \label{impactn:b}
		\includegraphics[width=0.24\columnwidth]{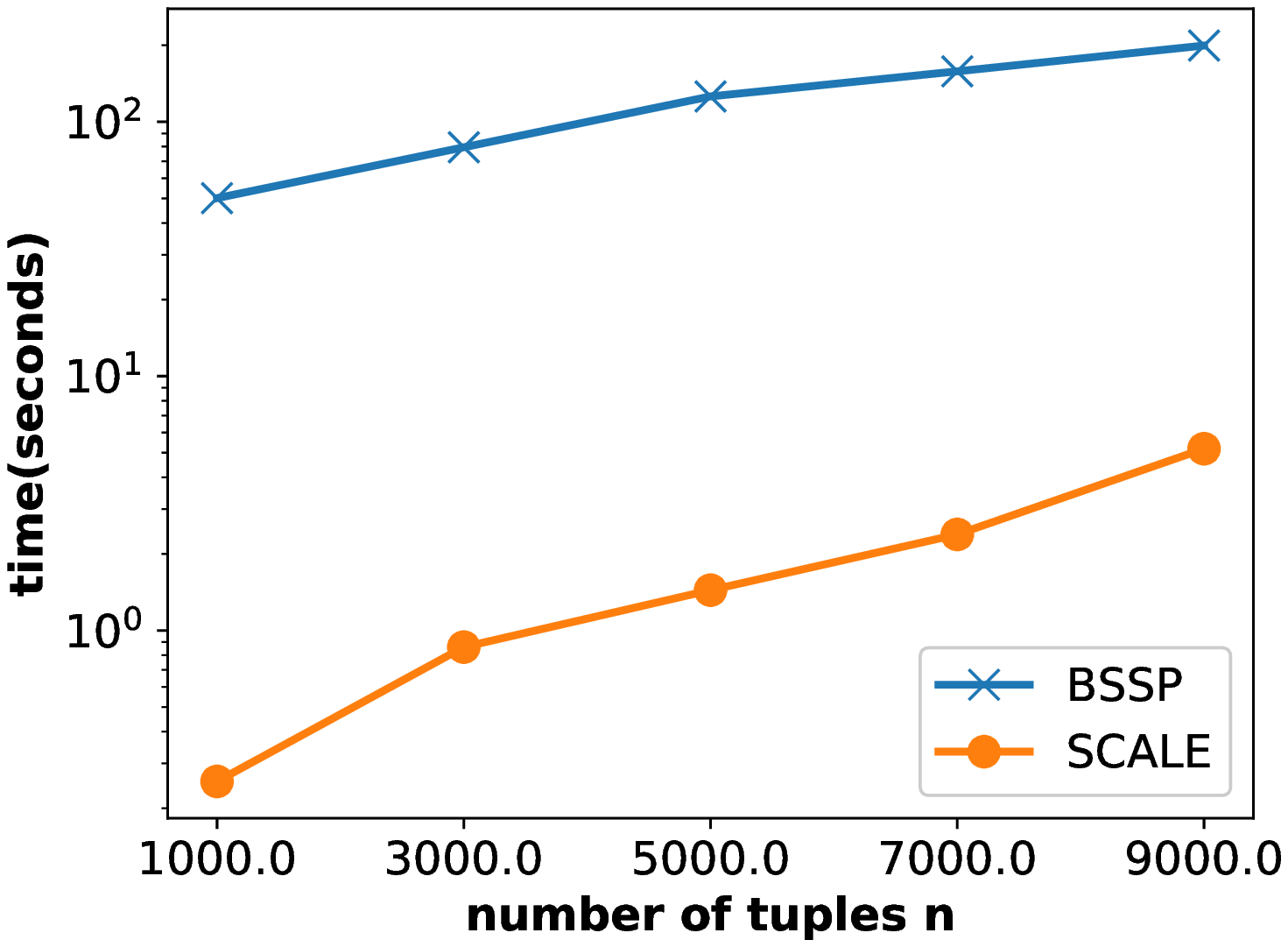}
	}
	\subfloat[INDE] { \label{impactn:c}
		\includegraphics[width=0.24\columnwidth]{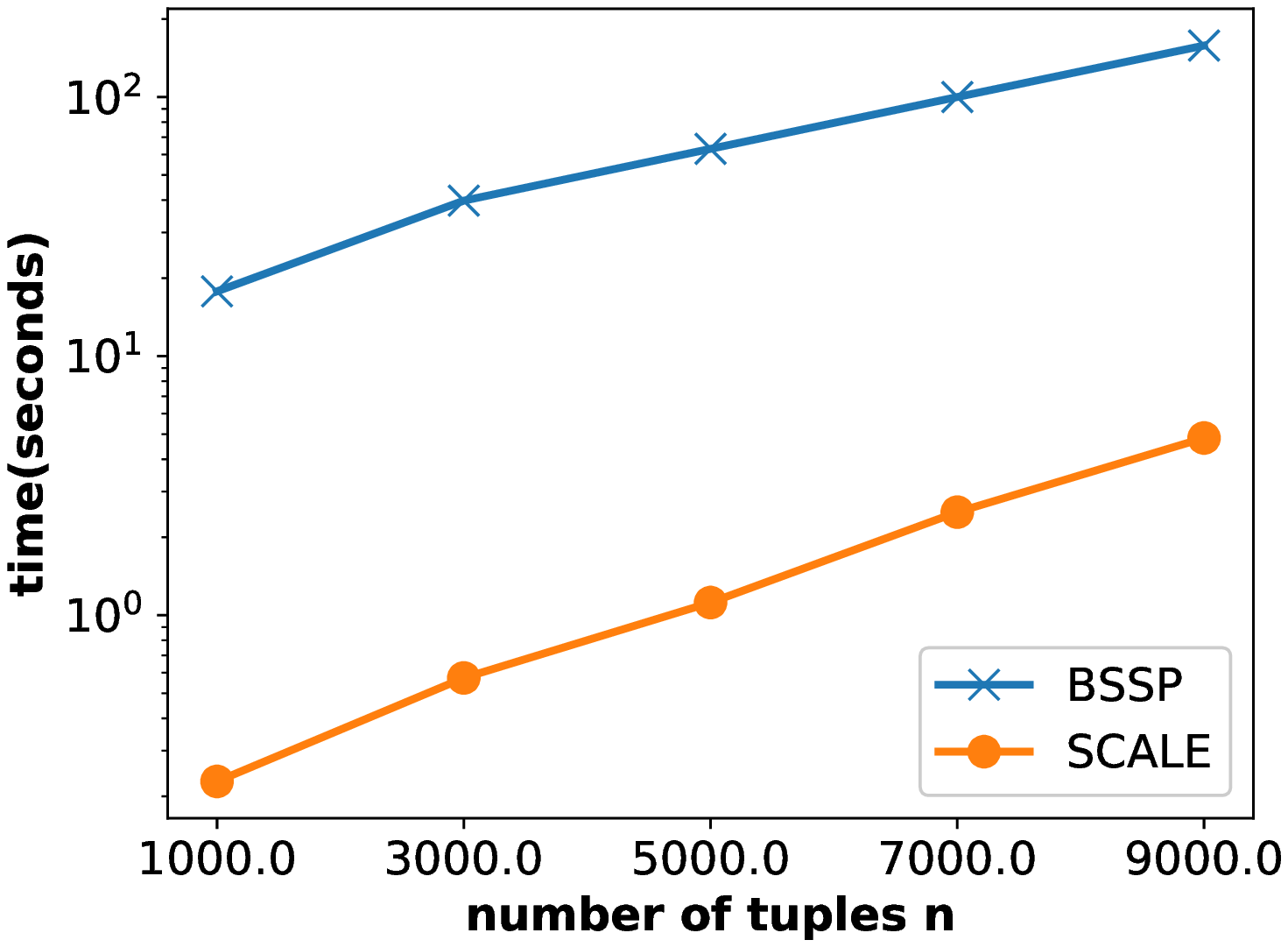}
	}
	\subfloat[NBA] { \label{impactn:d}
		\includegraphics[width=0.24\columnwidth]{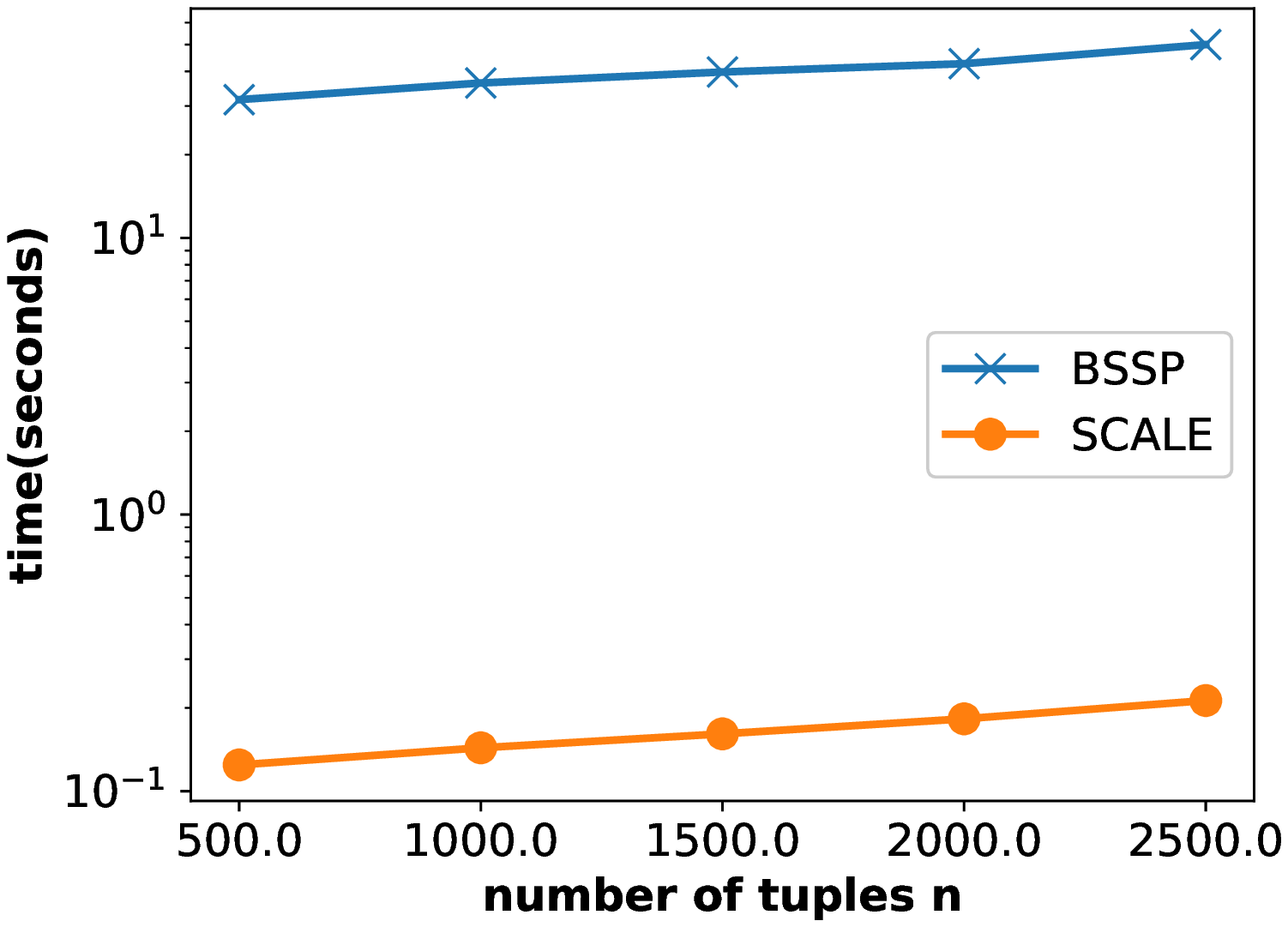}
	}
	\vspace{-1ex}\caption{Response time by varying the number of tuples (with $d=3, block=16, K=256$)}
	\label{impactn}
	\vspace{-2ex}\end{figure*}

%\begin{figure*}[htbp]
%		\centering
%\subfigure[CORR] { \label{impactm:a}
%\includegraphics[width=0.47\columnwidth]{test/test_for_dim_corr}
%}
%\subfigure[ANTI] { \label{impactm:b}
%\includegraphics[width=0.47\columnwidth]{test/test_for_dim_anti}
%}
%\subfigure[INDE] { \label{impactm:c}
%\includegraphics[width=0.47\columnwidth]{test/test_for_dim_inde}
%}
%\subfigure[NBA] { \label{impactm:d}
%\includegraphics[width=0.47\columnwidth]{test/test_for_dim_nba}
%}
%\caption{The Impact of $m$ ($n=2500, block=16, K=256$). }
%\label{impactm}
%\end{figure*}
%
%\begin{figure*}[htbp]
%		\centering
%\subfigure[CORR] { \label{impactblock:a}
%\includegraphics[width=0.47\columnwidth]{test/test_for_block_corr}
%}
%\subfigure[ANTI] { \label{impactblock:b}
%\includegraphics[width=0.47\columnwidth]{test/test_for_block_anti}
%}
%\subfigure[INDE] { \label{impactblock:c}
%\includegraphics[width=0.47\columnwidth]{test/test_for_block_inde}
%}
%\subfigure[NBA] { \label{impactblock:d}
%\includegraphics[width=0.47\columnwidth]{test/test_for_block_nba}
%}
%\caption{The Impact of ORE Block ($n=2500, m=3, K=256$).}
%\label{impactblock}
%\end{figure*}

\begin{figure*}[t]
	\centering
	\subfloat[Effect of $d$ (with $block=16, K=256$)]{ \label{impact:m}
		\includegraphics[width=0.3\columnwidth,height=3.1cm]{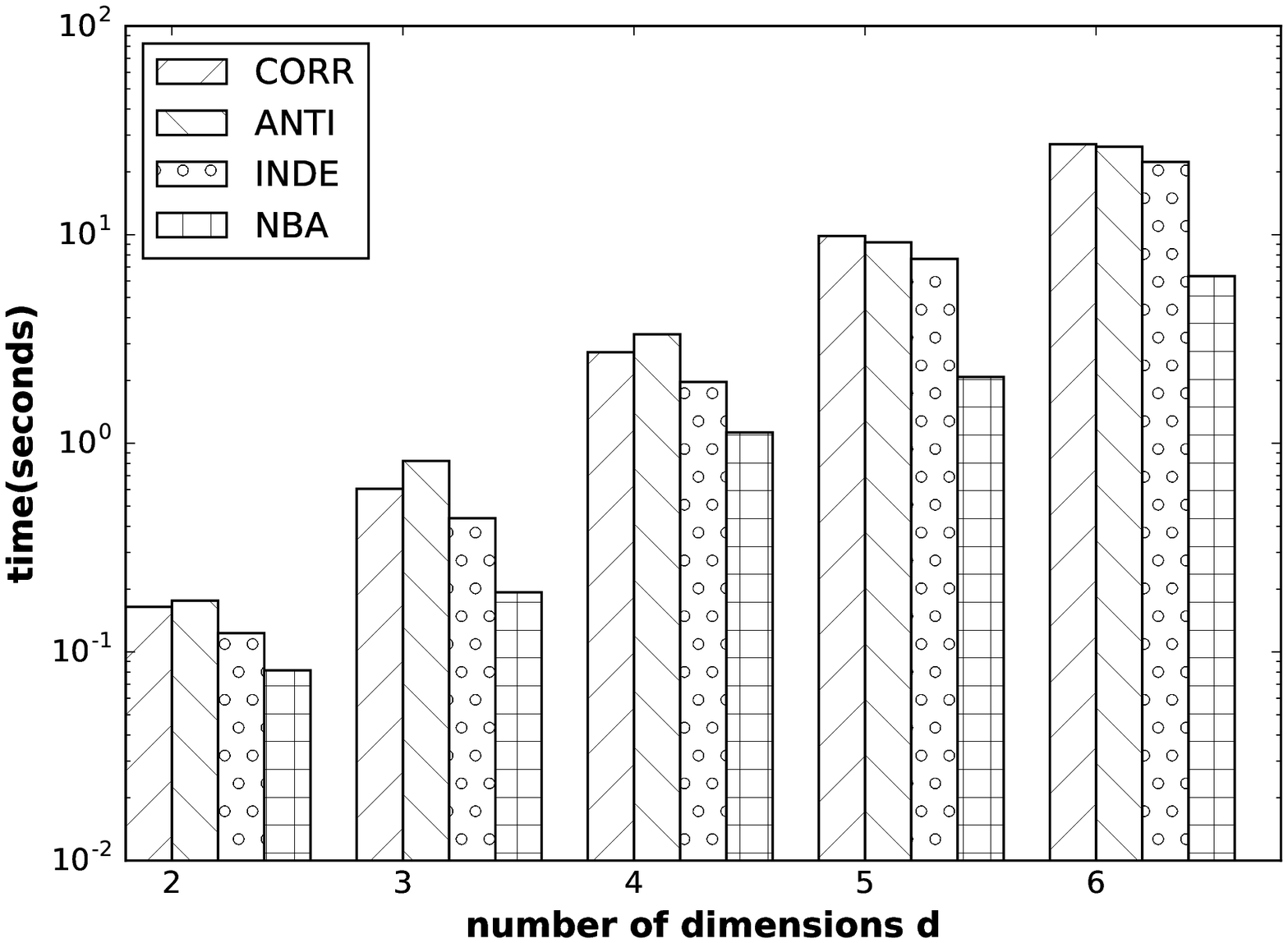}
	}\quad
	\subfloat[Effect of block size (with $d=3, K=256$)] { \label{impact:b}
		\includegraphics[width=0.3\columnwidth,height=3.1cm]{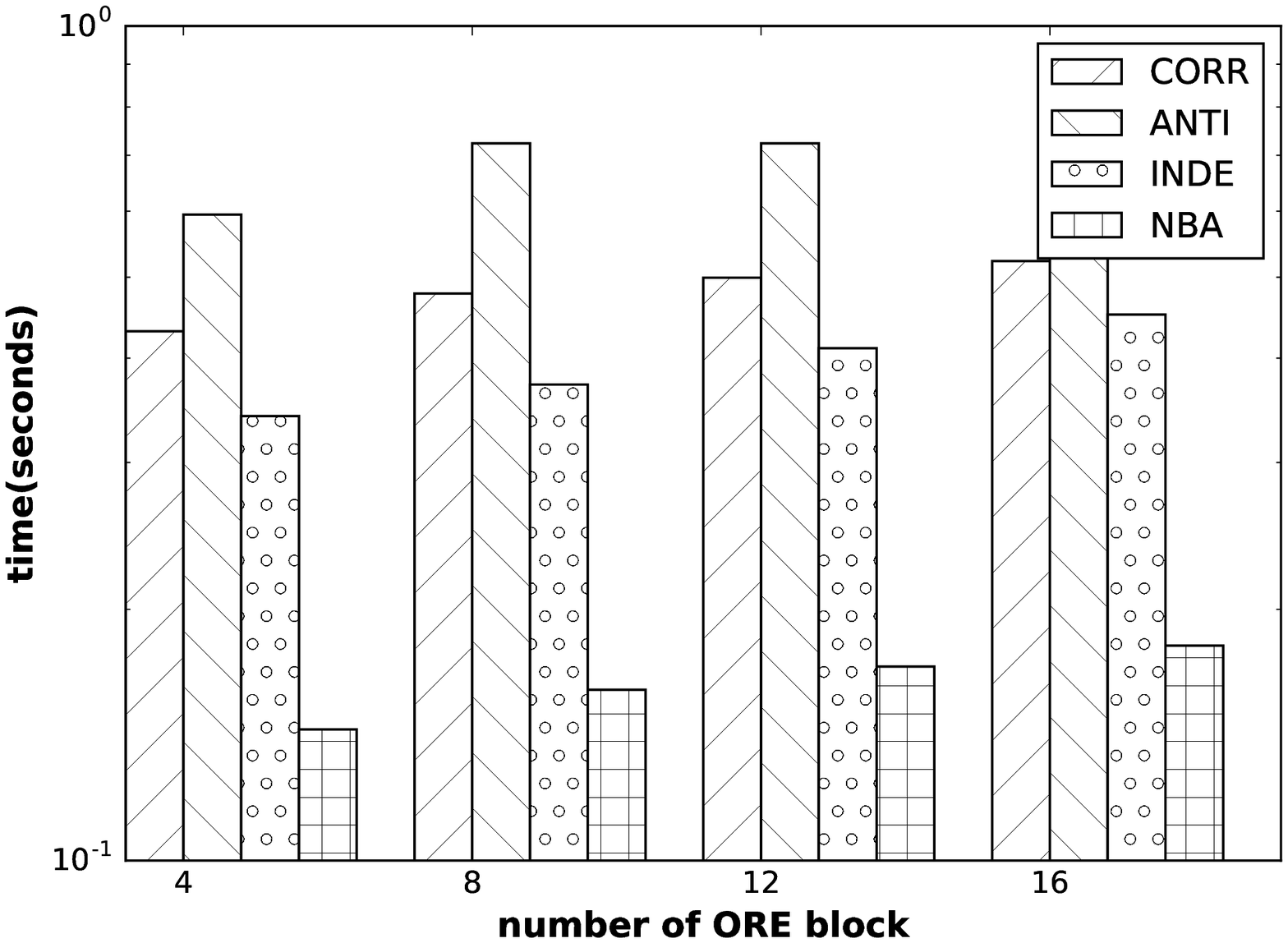}
	}\quad
	\subfloat[Effect of key length (with $block=16, d=3$)] { \label{impact:key}
		\includegraphics[width=0.3\columnwidth,height=3.1cm]{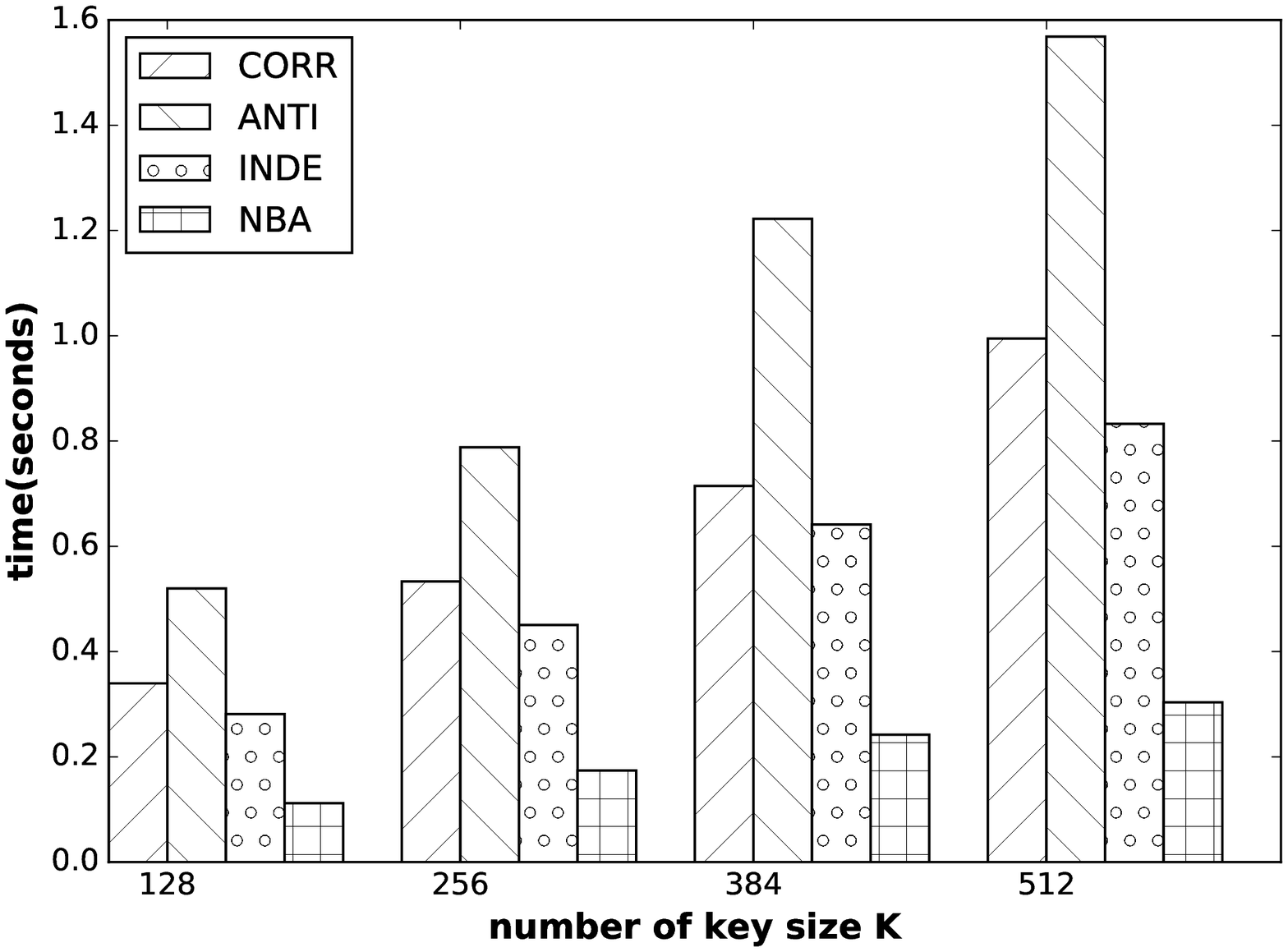}
	}
	\vspace{-1ex}\caption{The effects of different parameters ($n=2500$)}
	%\subfigure[NBA] { \label{impactkey:d}
	%\includegraphics[width=0.47\columnwidth]{test/test_for_key_nba}
	%}
	%\caption{The Impact of ORE Key Size $K$ ($n=2500, block=16, m=3$).}
	\label{impactall}
	\vspace{-3ex}
\end{figure*}

\begin{figure*}[t]
	\centering
	\subfloat[INDE] { \label{insert:list}
		\includegraphics[width=0.24\columnwidth]{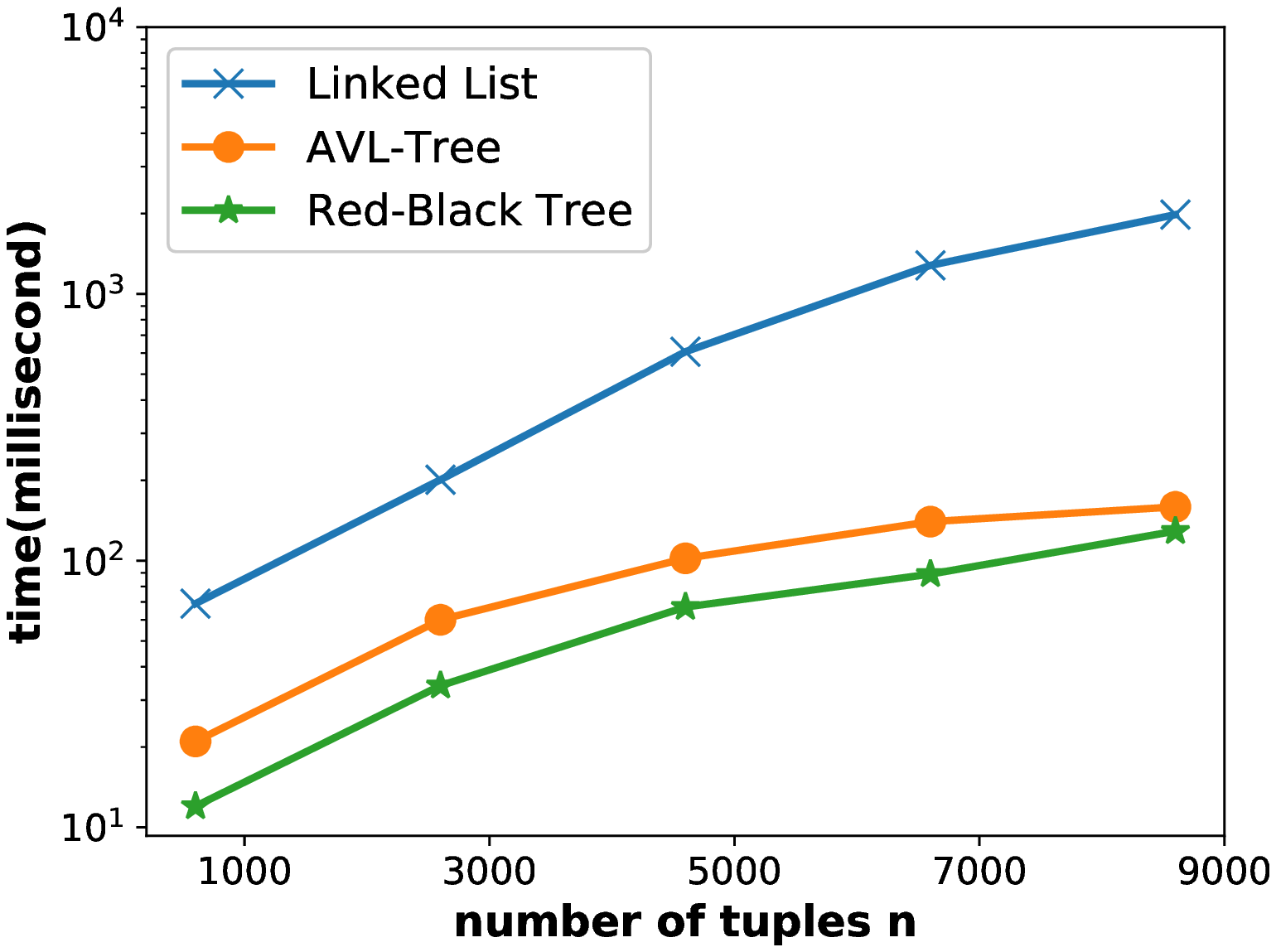}
	}
	\subfloat[CORR] { \label{insert:avl}
		\includegraphics[width=0.24\columnwidth]{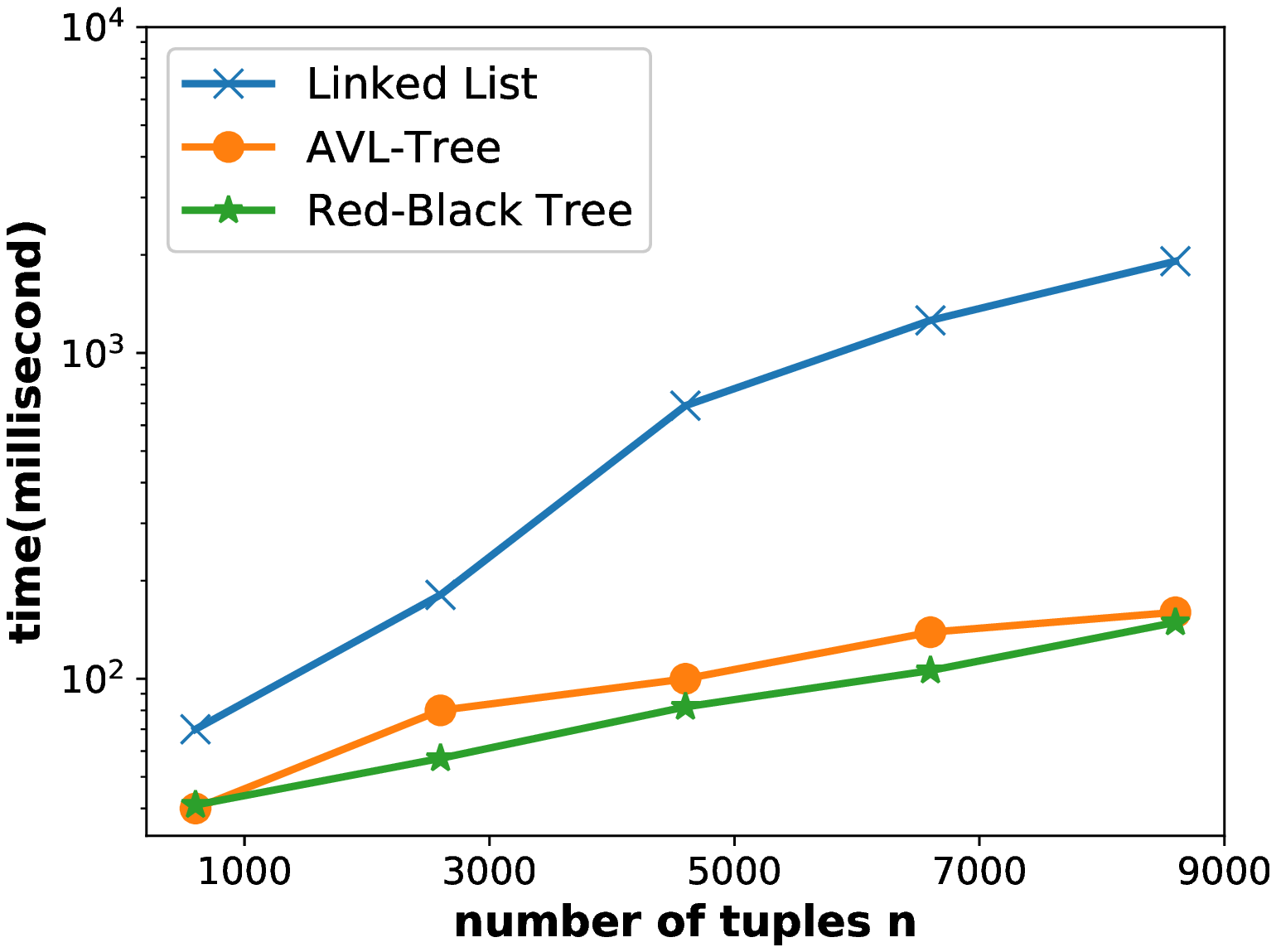}
	}
	\subfloat[ANTI] { \label{insert:rbt}
		\includegraphics[width=0.24\columnwidth]{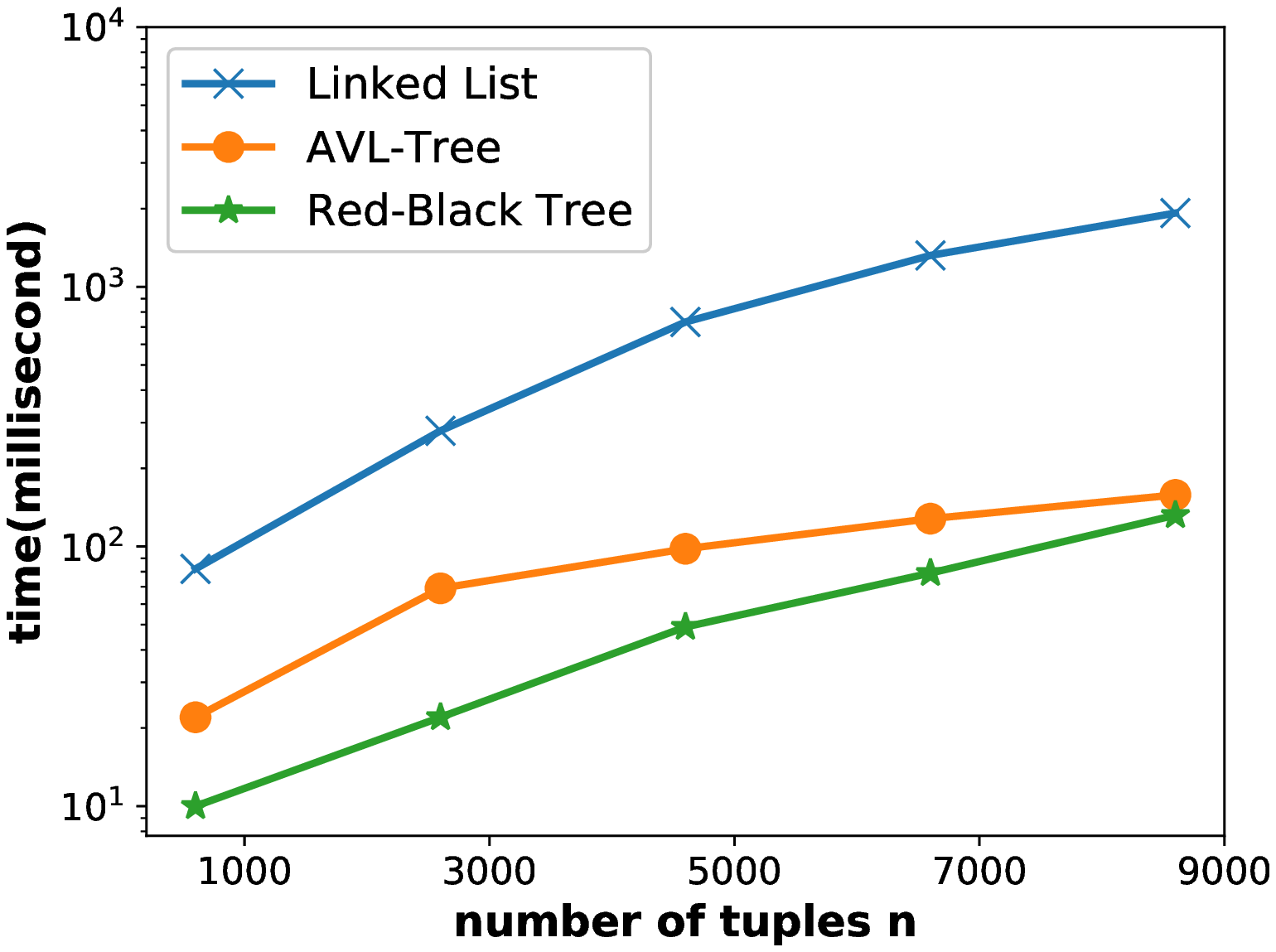}
	}
	\subfloat[NBA] { \label{insert:nba}
		\includegraphics[width=0.24\columnwidth]{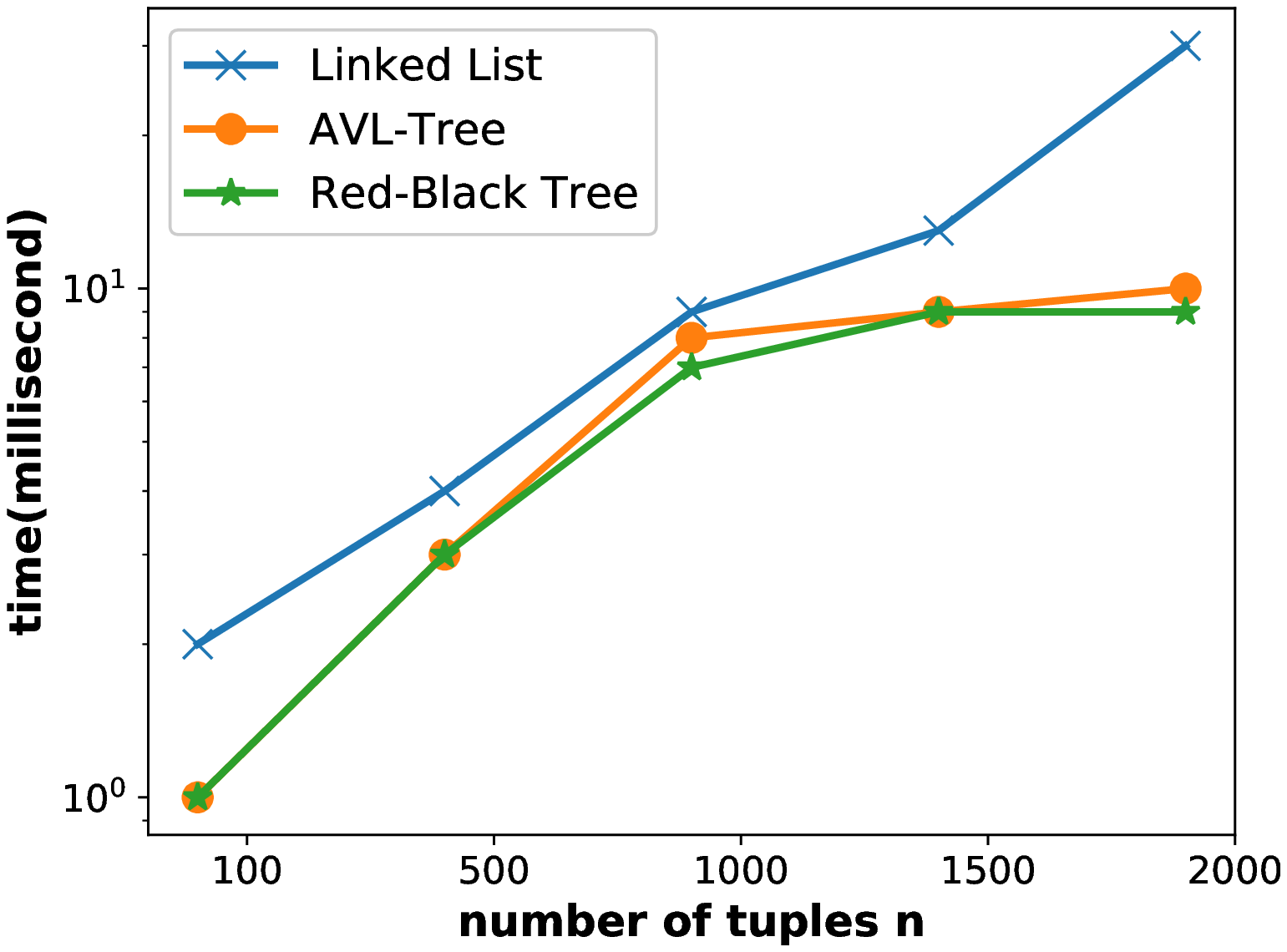}
	}
	\vspace{-2ex}\caption{Performance of different indexing structure during insertion}
	%\subfigure[NBA] { \label{impactkey:d}
	%\includegraphics[width=0.47\columnwidth]{test/test_for_key_nba}
	%}
	%\caption{The Impact of ORE Key Size $K$ ($n=2500, block=16, m=3$).}
	\label{insertData}
	\vspace{-2ex}
\end{figure*}

\begin{figure*}[t]
	\centering
	\subfloat[INDE] { \label{delete:list}
		\includegraphics[width=0.24\columnwidth]{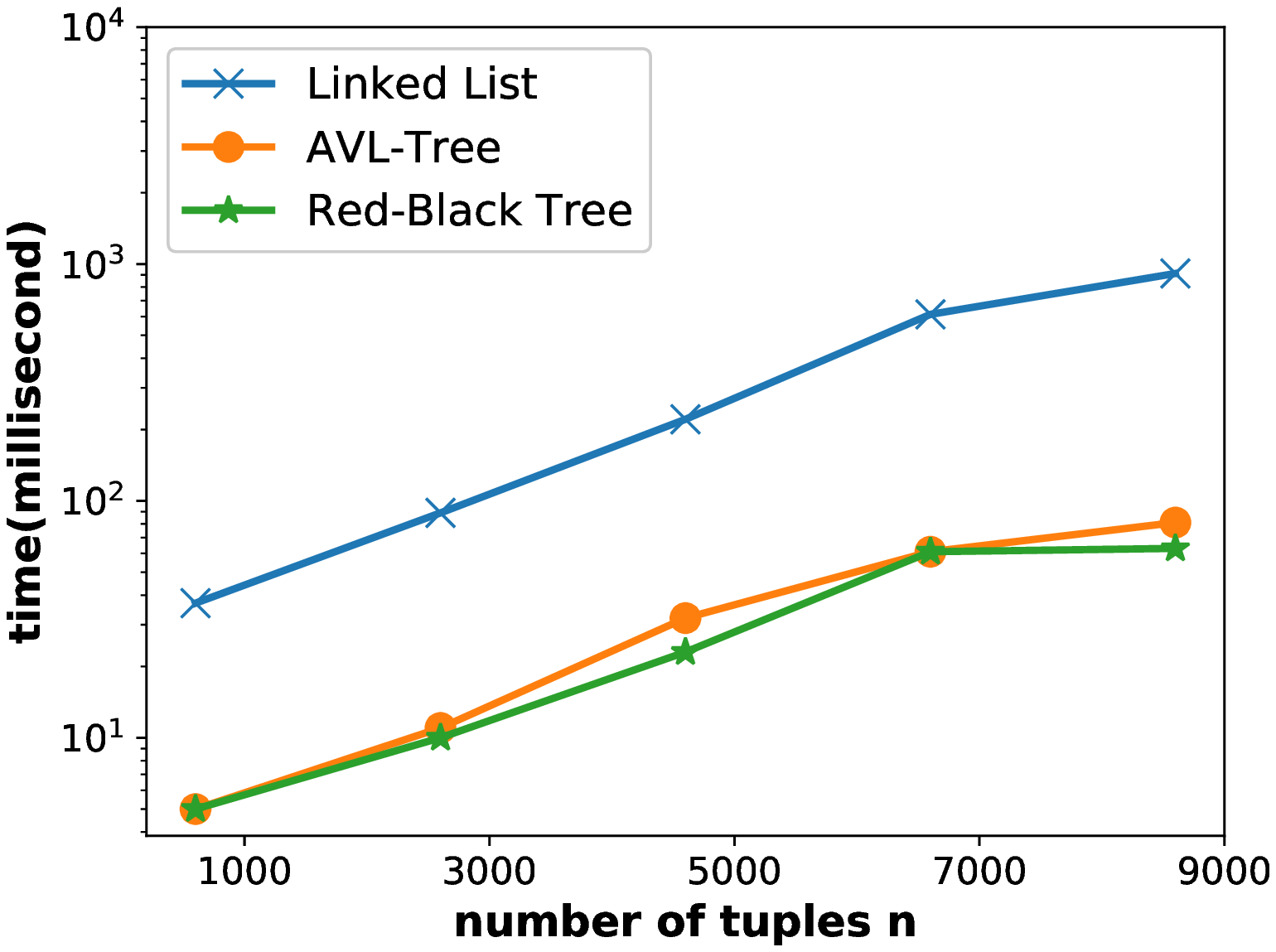}
	}
	\subfloat[CORR] { \label{delete:avl}
		\includegraphics[width=0.24\columnwidth]{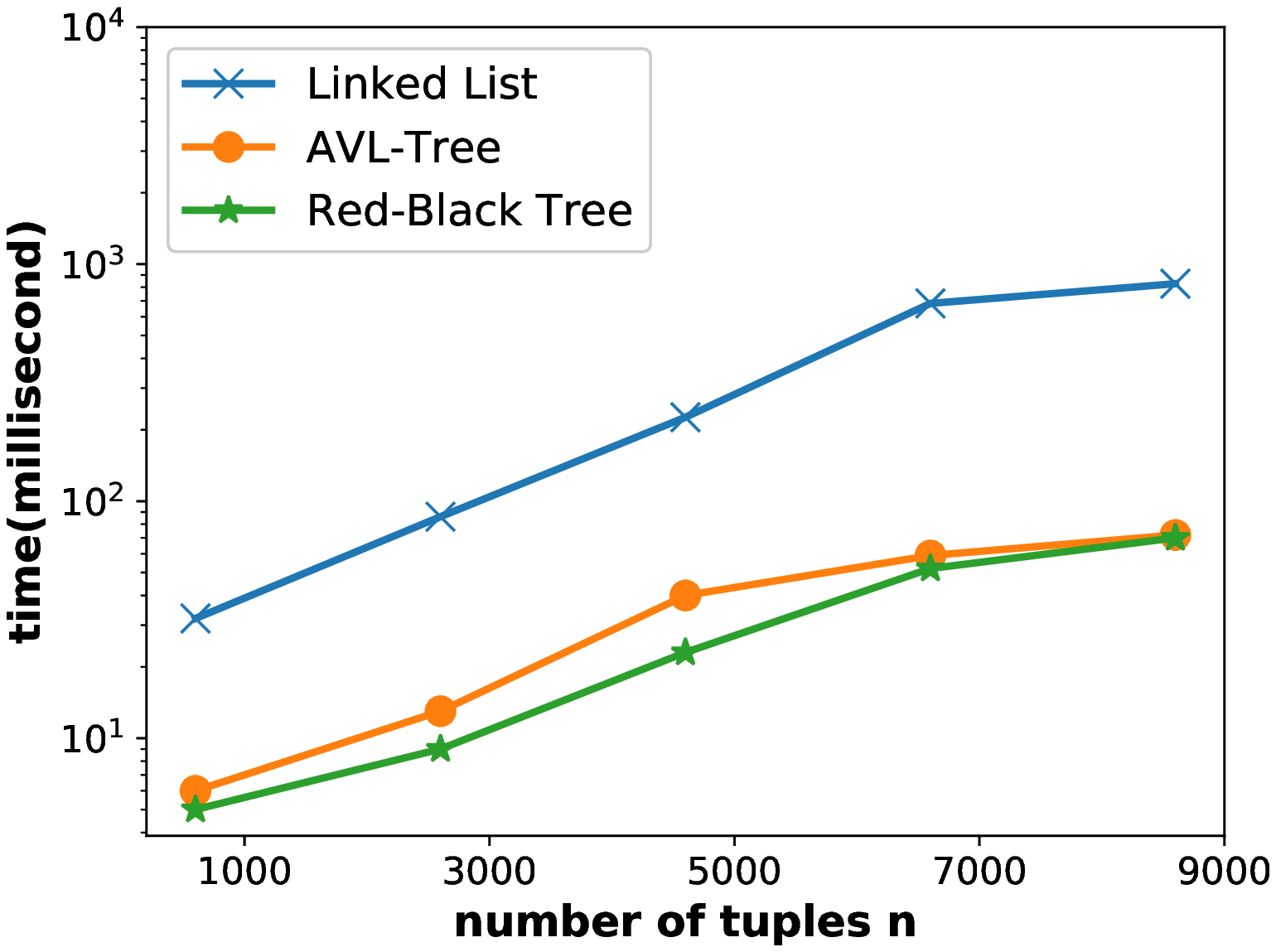}
	}
	\subfloat[ANTI] { \label{delete:rbt}
		\includegraphics[width=0.24\columnwidth]{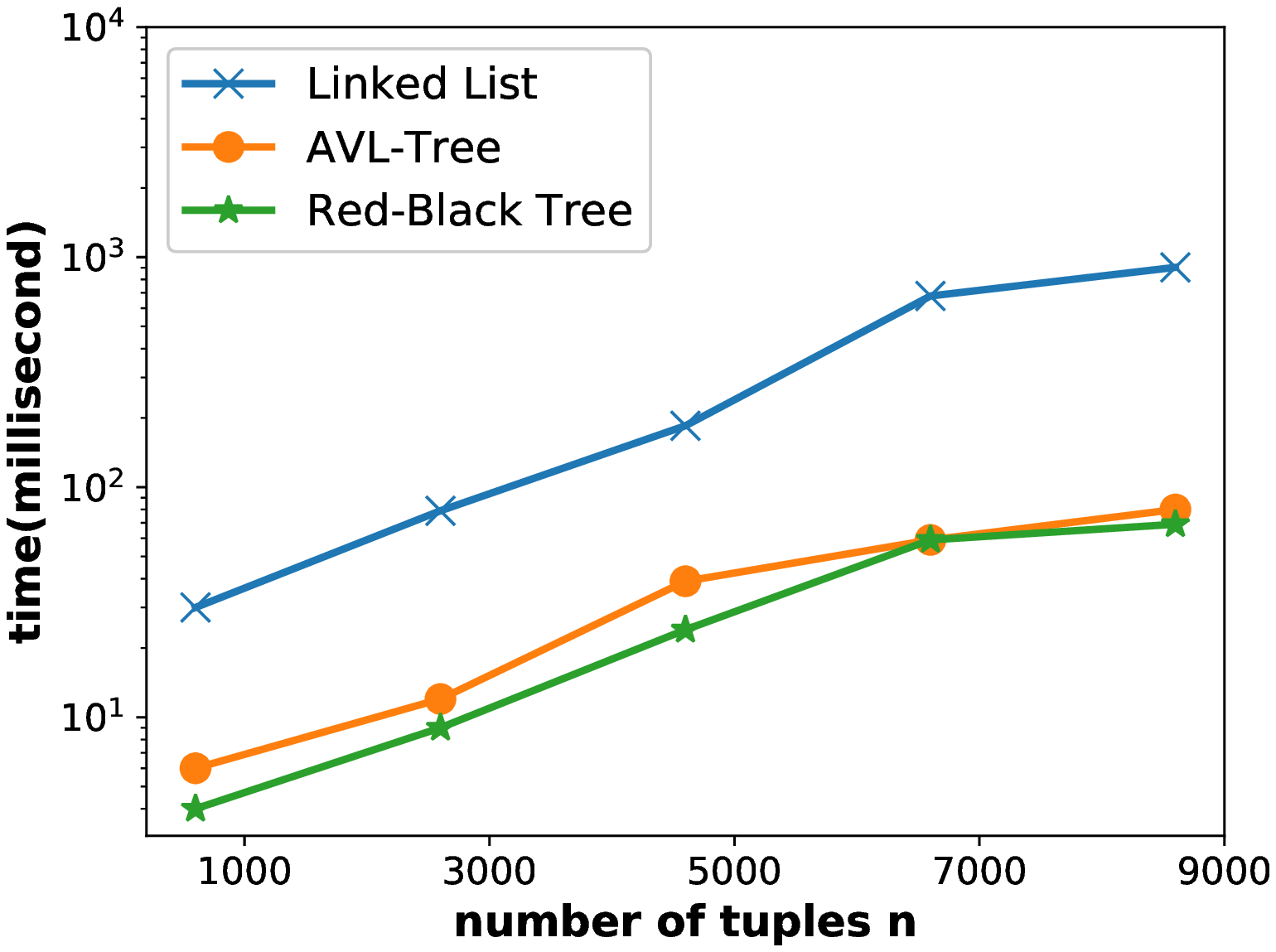}
	}
	\subfloat[NBA] { \label{delete:nba}
		\includegraphics[width=0.24\columnwidth]{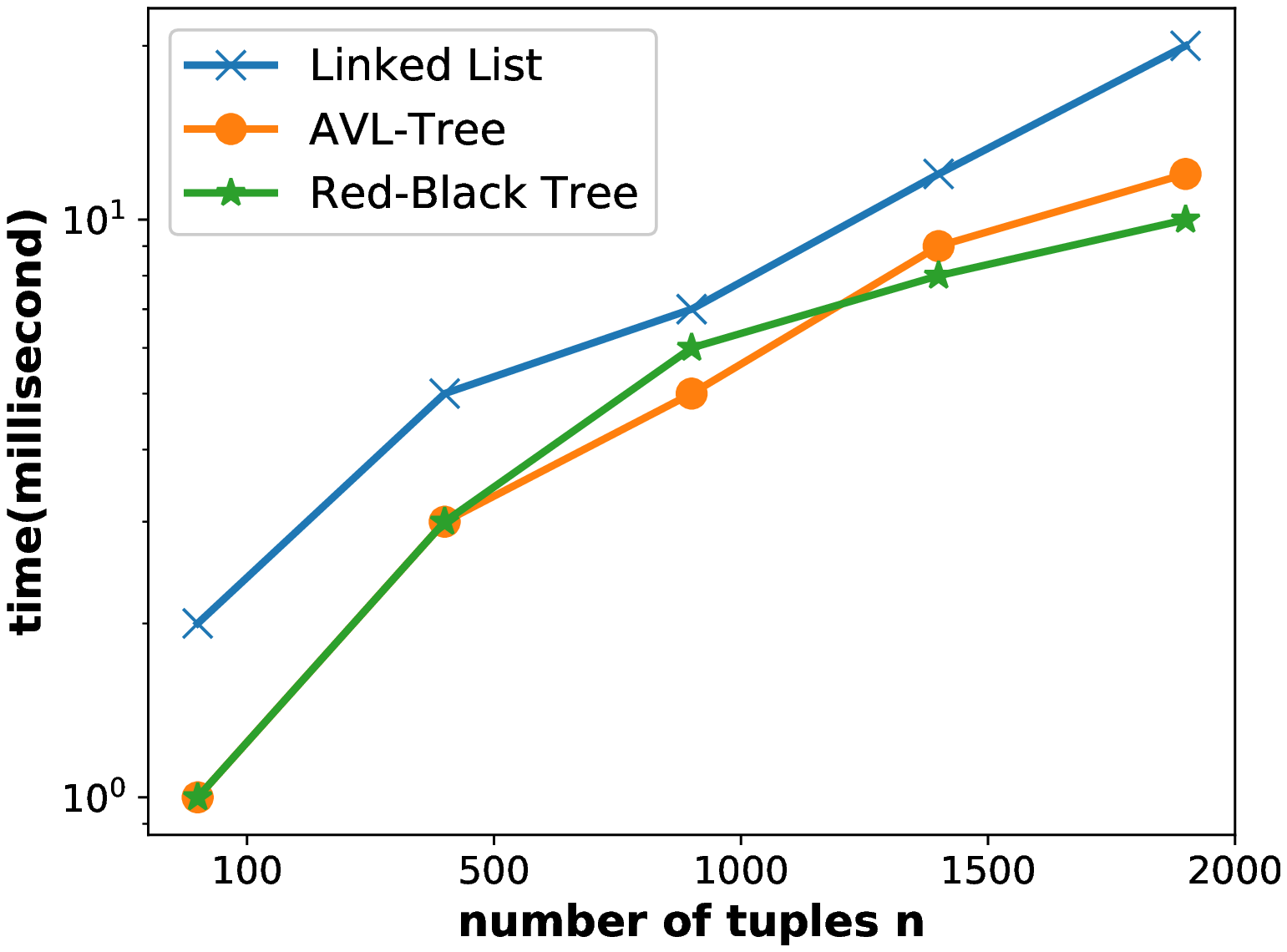}
	}
	\vspace{-1ex}\caption{Performance of different indexing structure during deletion}
	%\subfigure[NBA] { \label{impactkey:d}
	%\includegraphics[width=0.47\columnwidth]{test/test_for_key_nba}
	%}
	%\caption{The Impact of ORE Key Size $K$ ($n=2500, block=16, m=3$).}
	\label{deleteData}
	\vspace{-2ex}
\end{figure*}

\vspace{-1ex}\section{Experimental Study}
\label{sec:Experiments}

In this section, we evaluate the performance and scalability of \textsc{scale} under different parameter settings over four datasets, including both real-world and synthetic ones. We also compare our model with another baseline, namely BSSP~\cite{Liu2017Secure,DBLP:journals/corr/abs-1806-01168}, which is the only solution for secure dynamic skyline query.

\vspace{-1ex}\subsection{Experiment Settings}

All algorithms are implemented in C, and tested on the platform with a 2.7GHz Intel Core i5 processor and 8GB of memory running MacOS. We use both synthetic datasets and a real-world dataset in our experiments. In particular, we generated independent (INDE), correlated (CORR), and anti-correlated (ANTI) datasets following the seminal work in \cite{borzsony2001skyline}. In line with~\cite{Liu2017Secure,DBLP:journals/corr/abs-1806-01168}, we also adopt a dataset that contains 2500 NBA players who are league leaders of playoffs\footnote{It is acquired from \url{https://stats.nba.com/alltime-leaders/?SeasonType=Playoffs}.}. Each player is associated with six attributes that measure the player's performance: Points, Offensive Rebounds, Defensive Rebounds, Assists, Steals, and Blocks.

\vspace{-1ex}\subsection{Performance Results}
In this subsection, we evaluate our protocols by varying the number of tuples ($n$), the number of dimensions ($d$), the ORE block setting, and the length of key ($K$).

\btitle{Varying the number of tuples.} Fig.~\ref{impactn} shows the time cost by varying the number of tuples, namely $n$, on the four datasets. In this group of experiments, we fix the number of dimensions, ORE block size and key length as 3, 16 and 256, respectively. We observe that for all datasets, the time cost increases almost linearly with respect to $n$. This phenomenon is consistent with our complexity study in Section~\ref{ssec:compcompana}. Notably, for the real-world dataset (\ie NBA), the query response time is less than $0.2$ seconds, which is efficient enough in practice. Compared to the state-of-the-art~\cite{Liu2017Secure}, \textsc{scale} is more than 3 orders of magnitude faster. In the following, we shall fix $n=2500$ and focus on evaluating the effects of the three parameters in our scheme.

\btitle{Impact of $d$.} Fig.~\ref{impact:m} shows the time cost for different $d$ on the four datasets, where we fix the ORE block size and key length as 16 and 256, respectively. For all datasets, as $d$ increases from 2 to 6, the response time in all four datasets increases almost exponentially as well. This fact is consistent with the ordinary dynamic skyline querying in plaintext. This is because an increase in $d$ leads to more comparison operations for the decision of dynamic dominance criteria. \eat{The time for NBA dataset is lower than the CORR dataset which suggests that the NBA data is strongly correlated.}

\btitle{Impact of ORE block.} Encrypting plaintext based on block cipher, different block sizes may take different time. Fig.~\ref{impact:b} plots the time cost by varying the block sizes used in the ORE scheme, where $d$ and $K$ are fixed as 3 and 256, respectively. As mentioned in~\cite{Lewi2016Order}, this ORE scheme leaks the first block of $\delta$-bits that differs, therefore, increasing the block size brings higher security. Observe that the response time increases slightly with respect to the size of ORE block. That indicates, higher security level in ORE has to sacrifice some response time.

\btitle{Impact of $K$.} Fig.~\ref{impact:key} shows the time cost by varying the lengths for the keys in the ORE scheme. This ORE scheme uses AES as the building block, therefore, increasing the encryption key size brings in higher security. Similar to that of block size, the response time also increases linearly with respect to the size of encryption keys. Comparing Fig.~\ref{impact:key} against \ref{impact:b}, we can make the following observations. First, increasing the security level for ORE scheme will definitely sacrifice some efficiency. Second, the key length in AES exhibits more significant impact on the efficiency comparing to that of ORE block size.

\btitle{Maintenance cost.} We extensively perform another group of experiments to test the performance by adopting different index structures for managing the encrypted pairs of sums whenever insertion or deletion occurs. In particular, we compare three different structures including Linked List, Red-black Tree and AVL-Tree. Notably, as~\cite{Liu2017Secure,DBLP:journals/corr/abs-1806-01168} does not support modifications over existing database records, this method is not taken into account within this group of experiments.

Firstly, for each of the datasets, we sequentially insert (encrypting and uploading) new records into the database and evaluate the average response time for insertion of the corresponding ciphertext into existing indices, which includes both the time for searching the targeted position for insertion within the Lists or Trees as well as re-balancing the Trees, if needed, afterwards. The average response times, which are taken from 10 independent runs, for inserting different number of records using the three index structures are shown in Fig.~\ref{insertData}. In all the test for each dataset and index structure, we fix the parameters as $d=3, block=16, K=256$. Through the results in all the datasets, it is obvious that Linked List performs the worst, with about one order of magnitude larger running time than the other two Tree-based competitors. Both AVL-Tree and Red-Black Tree exhibit similar performance, much better than the Linked List. All these phenomenon is consistent with the complexities for these structures shown in Table~\ref{tb:insertAndDeleteTC}.  

Secondly, given all the existing datasets encrypted and stored using \textsc{scale}, we further test the response time for deletion of arbitrary records. Similar to the previous experiments, the response time includes both the time for searching the targeted position for deletion within the Lists or Trees as well as re-balancing the Trees, if needed, afterwards. The average response times, which are taken from 10 independent runs, for deleting different number of records using the three index structures are shown in Fig.~\ref{deleteData}. We also fix the parameters as $d=3, block=16, K=256$ throughout all the datasets. According to the results in all the datasets, Linked List performs the worst again, with about one order of magnitude larger running time than the other two Tree-based competitors. Both AVL-Tree and Red-Black Tree exhibit similar performance. The cost for all the methods increase with respect to the number of deleted records.

Notably, throughout all the datasets in both Fig.~\ref{insertData} and~\ref{deleteData}, Red-Black Tree is slightly better than AVL-Tree for inserting new records or deleting existing ones, as it do not need to strictly balancing the tree after each insertion or deletion. However, unbalanced tree structure will inevitably lead to worse query performance. Hence, there is a trade-off between the maintenance cost and query cost in selecting a solution from the two different tree-based index structure. In particular, when query performance is more sensitive, AVL-Tree should be adopted (all other experiments in this section are conducting based on this method); when maintenance cost is more sensitive, Red-Black Tree should be a better choice.

\vspace{-1ex}\section{Conclusion}
\label{sec:Conclusion}
In this paper, we have presented a new framework called \textsc{scale} to address the secure dynamic skyline query problem in the cloud platform. A distinguishing feature of our framework is the conversion of the requirement of both subtraction and comparison operations to only comparisons. As a result, we are able to use ORE to realize dynamic domination protocol over ciphertext. Based on this feature, we built \textsc{scale} on top of BNL. In fact, our framework can be easily adapted to other plaintext dynamic skyline query models. We theoretically show that the proposed scheme is secure under our system model, and is efficient enough for practical applications. Moreover, there is only one interaction between user and the cloud, which minimizes the communication cost and corresponding threats. Besides, we also present a mechanism for modification over existing stored records. Based on the proposed mechanism, insertion, deletion and update of records can be all efficiently supported by \textsc{scale}. Experimental study over both synthetic and real-world datasets demonstrates that \textsc{scale} improves the efficiency with at least three orders of magnitude compared to the state-of-the-art method. As part of our future work, we plan to further enhance the security of our scheme and explore how the scheme can be adapted to support other variations for skyline query.

%\vspace{-1ex}\begin{acks}
%	The work is supported by National Nature Science Foundation of China (No. 61672408 and 61972309), Fundamental Research Funds for the Central Universities (No. JB181505), Natural Science Basic Research Plan in Shaanxi Province of China (No. 2018JM6073), National Engineering Laboratory for Public Safety Risk Perception and Control by Big Data (PSRPC), and China 111 Project (No. B16037).
%\end{acks}

\vspace{-1ex}\appendix
\section{Proof of Theorem~\ref{thm:keynum}}\label{apxa}
As shown in Fig.~\ref{add}, the relations for sums in the same class (\eg line) can be inferred easily from $E(P)$. Obviously, $n$ elements can be divided into $\kappa$ groups in the way shown in Fig.~\ref{add}. In detail, given $n$ elements, we can organize all the paired summations of them into a matrix, whose row/column correspond to each element and are arranged according to the descending order of the elements. Each entry is the sum of the corresponding row and column. As it is a symmetrical matrix, all the paired sums can be found in the upper-right corner without the diagonal. Obviously, all the entries along the borders (upper and right) of the corner belong to the same \textit{OOC}. By removing these entries from the matrix, we can iteratively find several inner borders (\ie \textit{OOC}s). Obviously, the first \textit{OOC} contains $2n-3$ entries, and each subsequent class has 4 entries less than previous one. Hence, the total number of classes is $\kappa=\lceil\frac{2*n-3}{4}\rceil$. Assuming that we use $\kappa-\epsilon$ secret keys for encryption, there must exist two \textit{OOC}s sharing the same key. Then, the cloud may infer the distribution of values in plaintext tuples in the aforementioned way. In this regard, besides the order, the distribution of the values in a dimension is also leaked, which deviates from the security model defined in Section~\ref{ssec:secmodel}. Therefore, to realized the predefined security requirements, the minimum number of keys for pairs of sums in a dimension should be at least $\kappa$.\EndOfProof

\vspace{-1ex}\section{Proof of Theorem~\ref{theorem2}}\label{apxb}
Let the leakage function collection be
$$\mathcal{L} = (\mathcal{L}_{Encrypt}, \mathcal{L}_{Query}, \mathcal{L}_{Insert}, \mathcal{L}_{Delete}).$$

\noindent\emph{Assumption.} The real game ${\rm Game}_{\mathcal{R},\mathcal{A}}$ strictly follows the presented \textsc{scale} construction. In the ideal game ${\rm Game}_{\mathcal{S},\mathcal{A}}$, hush functions are replaced by random oracles. Concretely, $\mathcal{S}$ maintains several hash tables for the adopted ORE scheme and another hash table consisting of tuples $(in, out)\in (\{0,1\}^{*} \times \{0,1\}^{*}) \times \{0,1\}^{\lambda}$. The first hash tables work as~\cite{Lewi2016Order} states, and the second hash table works as shown in Algorithm~\ref{Alg:RO}. Specifically, $h$ in Section~\ref{section:AccessingPair} is replaced by $RO$ which works as shown in Algorithm~\ref{Alg:RO}. The rest part of ideal secure skyline query follows the presented \textsc{scale} construction.

\begin{algorithm}[h]
	\footnotesize
	\caption{Work flow of random oracle (RO).}
	\label{Alg:RO}
	\begin{algorithmic}[1]
		\REQUIRE $in$
		\ENSURE $out$
		\IF{$\exists (in_i, out_i) \in RO$ such that $in_i=in$}
		\STATE Return $out_i$;
		\ELSE {}
		\STATE $out \xleftarrow{R} \{0, 1\}^{\lambda}$;
		\STATE Add $(in, out)$ into $RO$;
		\STATE Return $out$;
		\ENDIF
	\end{algorithmic}
\end{algorithm}

\noindent\emph{Analysis.} Notably, the leakage function of the adopted ORE scheme is $\mathcal{L}_{BLK}$. Then, by introducing the conclusion in~\cite{Lewi2016Order}, the advantage for any probabilistic polynomial-time adversary breaking ORE is
\begin{equation*}
	\begin{split}
		|{\rm Pr}&[{\rm Game}_{\mathcal{R},\mathcal{A}(\lambda)}^{ORE}=1]-{\rm Pr}[{\rm Game}_{\mathcal{S}, \mathcal{A}(\lambda)}^{ORE} = 1]|\leq {\rm negl}^{ORE}(\lambda).
	\end{split}
\end{equation*}

By convention, the security and information leakage is derived fourfold as follows. In the following convention, we begin by the case that $p_i$ is one-dimensional. Then, the conclusion for multi-dimensional can be facilely generalized since the secret keys for each dimension are completely independently selected.

\btitle{Encrypt$_{\mathcal{S}}$}. In the process of database encryption, the database is originally sorted in ascending order. 

On one hand, for any $p_i,p_j \in P$ and $i \neq j$, the result of $p_i+p_j$ is encrypted by ORE scheme. During the encryption, $\lceil\frac{2*n-3}{4}\rceil$ secret keys are independently selected to minimize the information leakage. By convention of the encryption for pairs of tuples as illustrated in Fig.~\ref{Indexes}, there are no valid encrypted pairs of tuple such that $p_i+p_j<p_k+p_{\ell}$ and $i< k< j <\ell$. In such case, for each \textit{OOC} of encrypted pairs of tuples, there is no effective algorithm to infer any distributions in a single \textit{OOC}. The leakage function for $t$-th \textit{OOC} is $\mathcal{L}_{BLK}({X})$ such that $X=\{\mathsf{Enc}(p_t+p_i)|t<i<n-t+1\}\cup\{\mathsf{Enc}(p_i+p_{(n-t+1)})|t<i<n-t+1\}$. Additionally, all \textit{OOC}s are encrypted with independent secret keys. Hence, the complete leakage function is $\mathcal{L}_{BLK}(X)$, where $X=\cup_{t=1}^{(\lceil\frac{2*n-3}{4}\rceil)}(Y_t)$ and $Y_t=\{\mathsf{Enc}(p_t+p_j)|t<j<n-t+1\}\cup\{\mathsf{Enc}(p_j+p_{(n-t+1)})|t<j<n-t+1\}$. Because, such a leakage cannot put any effect on breaking ORE, the advantage for any probabilistic polynomial-time adversary breaking \textsc{scale} from this point of view is less than or equal to ${\rm negl}^{ORE}(\lambda)$.

On the other hand, the searching process is speedup by introducing $h$ as shown in Section~\ref{section:AccessingPair}. In the simulation, $h$ is replaced by $RO$. The input of $RO$ is $\mathsf{Enc}(p_i+p_j)$, and the output is a random number with $\lambda$ bits. If there exist $k$ and $\ell$ such that $p_k+p_{\lambda}=p_i+p_j$, the probability for $RO(\mathsf{Enc}(p_i+p_j)) \neq RO(\mathsf{Enc}(p_i+p_j))$ is at most ${\rm poly}(\lambda)/2^{\lambda}$. The largest \textit{OOC} contains $2n-3$ ciphertexts, so the probability of breaking \textsc{scale} from this perspective is at most $(2n-3){\rm poly}(\lambda)/2^{\lambda}$.

 The probability of breaking \textsc{scale} is ${\rm negl}^{ORE}(\lambda) + (2n-3){\rm poly}(\lambda)/2^{\lambda}$ during database encryption. The leakage function is $\mathcal{L}_{BLK}(X)$, where $X=\cup_{t=1}^{(\lceil\frac{2*n-3}{4}\rceil)}(Y_t)$ and $Y_t=\{\mathsf{Enc}(p_t+p_j)|t<j<n-t+1\}\cup\{\mathsf{Enc}(p_j+p_{(n-t+1)})|t<j<n-t+1\}$.

\btitle{Query$_{\mathcal{S}}$}. During query, according to the query request encryption as shown in Algorithm~\ref{encryptquery}, for each chain of encrypted pairs of tuples, $q$ and $2q$ are encrypted and temporarily inserted into the chain. Since cloud is assumed be semi-honest and will store all temporary results (\ie $\mathsf{Enc}(q)$ and $\mathsf{Enc}(2q)$), then the leakage function is  $\mathcal{L}_{BLK}({X}\cup \mathsf{Enc}(q) \cup \mathsf{Enc}(2q))$ such that $X=\{\mathsf{Enc}(p_t+p_i|t<i<n-t+1)\}\cup\{\mathsf{Enc}(p_i+p_{(n-t+1)})|t<i<n-t+1\}$. Hence, the complete leakage function is $\mathcal{L}_{BLK}(X)$, where $X=\cup_{i=1}^{(\lceil\frac{2*n-3}{4}\rceil)}(Y_t\cup \mathsf{Enc}_t(q) \cup \mathsf{Enc}_t(2q))$, $Y_t=\{\mathsf{Enc}(p_t+p_j)|t<j<n-t+1\}\cup\{\mathsf{Enc}(p_j+p_{(n-t+1)})|t<j<n-t+1\}$ and the advantage for any probabilistic polynomial-time adversary breaking \textsc{scale} from this point of view is still $\le{\rm negl}^{ORE}(\lambda)$. Additionally, from the view of random oracle, due to that there is no additional data that are inserted into the ciphertext storing structure, the probability of breaking \textsc{scale} from this point is still at most $(2n-3){\rm poly}(\lambda)/2^{\lambda}$. In a word, during query, the probability of breaking \textsc{scale} is ${\rm negl}^{ORE}(\lambda) + (2n-3){\rm poly}(\lambda)/2^{\lambda}$, and the leakage function is $\mathcal{L}_{BLK}(X)$, where $X=\cup_{i=1}^{(\lceil\frac{2*n-3}{4}\rceil)}(Y_t\cup \mathsf{Enc}_t(q) \cup \mathsf{Enc}_t(2q))$ and $Y_t=\{\mathsf{Enc}(p_t+p_j)|t<j<n-t+1\}\cup\{\mathsf{Enc}(p_j+p_{(n-t+1)})|t<j<n-t+1\}$.

\btitle{Insert$_{\mathcal{S}}$}. During insertion, there will be $n$ encrypted pairs of tuples that are inserted into the chains for only a single tuple. The conclusion can be derived by deducing the conclusion for \textbf{Encrypt$_{\mathcal{S}}$}. During insertion, the probability of breaking \textsc{scale} is ${\rm negl}^{ORE}(\lambda) + (2n-1){\rm poly}(\lambda)/2^{\lambda}$, and the leakage function is $\mathcal{L}_{BLK}(X)$, where $X=\cup_{t=1}^{(\lceil\frac{2*n-1}{4}\rceil)}(Y_t)$ and $Y_t=\{\mathsf{Enc}(p_t+p_j)|t<j<n-t+2\}\cup\{\mathsf{Enc}(p_j+p_{(n-t+2)})|t<j<n-t+2\}$.

\btitle{Delete$_{\mathcal{S}}$}. During deletion, there will be $n$ encrypted pairs of tuples that are removed from the chains for each single tuple. However, a semi-honest cloud may not remove such pairs. So, the probability of breaking \textsc{scale} and the leakage function should be the same as that in \textbf{Encrypt$_{\mathcal{S}}$}.
Therefore, for one-dimensional data, we can conclude that the leakage function during different phases are as follows.
\begin{align*}
	\mathcal{L}_{Encrypt}= \mathcal{L}_{BLK}(X^{(n)}),
	\mathcal{L}_{Query}= \mathcal{L}_{BLK}(X^{(n)'}),\\
	\mathcal{L}_{Insert}= \mathcal{L}_{BLK}(X^{(n+1)}),
	\mathcal{L}_{Delete}= \mathcal{L}_{BLK}(X^{(n)})
\end{align*}	
where $X^{(n)}=\cup_{t=1}^{(\lceil\frac{2*n-3}{4}\rceil)}(Y_t)$, $X^{(n)'}=\cup_{t=1}^{(\lceil\frac{2*n-3}{4}\rceil)}(Y_t\cup \mathsf{Enc}_t(q) \cup \mathsf{Enc}_t(2q))$ and $Y_t=\{\mathsf{Enc}(p_t+p_j)|t<j<n-t+1\}\cup\{\mathsf{Enc }(p_j+p_{(n-t+1)})|t<j<n-t+1\}$. Then, the advantage for breaking \textsc{scale} is at most ${\rm negl}^{ORE}(\lambda) + (2n-1){\rm poly}(\lambda)/2^{\lambda}$.

In fact, each record may have $d$ dimensions. The secret keys for each dimension are selected independently, so the leakage function is the union of leakage for each dimension, and the advantage for breaking \textsc{scale} is $d$ times larger than that for each dimension.

\noindent\emph{Conclusion.} The leakage function in different phases are as follows.
\begin{align*}
	\mathcal{L}_{Encrypt}= \mathcal{L}_{BLK}(\cup_{k=1}^{d}X_k^{(n)}),
	\mathcal{L}_{Query}= \mathcal{L}_{BLK}(\cup_{k=1}^{d}X_k^{(n)'}),\\
	\mathcal{L}_{Insert}= \mathcal{L}_{BLK}(\cup_{k=1}^{d}X_k^{(n+1)}),
	\mathcal{L}_{Delete}= \mathcal{L}_{BLK}(\cup_{k=1}^{d}X_k^{(n)})
\end{align*}		
where $X_k^{(n)}=\cup_{t=1}^{(\lceil\frac{2*n-3}{4}\rceil)}(Y_t^{k})$, $X_{k}^{(n)'}=\cup_{t=1}^{(\lceil\frac{2*n-3}{4}\rceil)}(Y_t^{k}\cup \mathsf{Enc}_t(q) \cup \mathsf{Enc}_t(2q))$ and $Y_t^k=\{\mathsf{Enc}(p_t[k]+p_j[k])|t<j<n-t+1\}\cup\{\mathsf{Enc}(p_j[k]+p_{(n-t+1)}[k])|t<j<n-t+1\}$. Therefore, the advantage for breaking \textsc{scale} is at most $d\cdot({\rm negl}^{ORE}(\lambda) + (2n-1){\rm poly}(\lambda)/2^{\lambda})$.\EndOfProof

%%
%% The acknowledgments section is defined using the "acks" environment
%% (and NOT an unnumbered section). This ensures the proper
%% identification of the section in the article metadata, and the
%% consistent spelling of the heading.

%\bibliographystyle{IEEEtran}
%\bibliography{IEEEabrv,secureskyline}
\bibliographystyle{splncs04}
\bibliography{secureSkyline}

\end{document}